\pdfoutput=1
\documentclass[nonacm,manuscript,acmsmall]{acmart}
\setcopyright{acmlicensed}
\copyrightyear{2025}
\acmYear{2025}
\acmDOI{XXXXXXX.XXXXXXX}

\acmJournal{TOPLAS}

\usepackage{newtxtext,newtxmath}
\usepackage{amsmath, amsfonts}
\usepackage[T1]{fontenc}
\usepackage{graphicx}
\usepackage[utf8]{inputenc}
\usepackage{hyperref}
\usepackage{color}

\usepackage{listings}
\usepackage{xcolor}
\usepackage{mathtools}
\usepackage{stmaryrd}
\usepackage{microtype}
\usepackage{multicol}
\usepackage{soul}
\usepackage[capitalize,nameinlink]{cleveref}
\usepackage[inline]{enumitem}
\usepackage{xspace}
\usepackage{stackrel}
\usepackage[all]{xy}
\usepackage{orcidlink}

\usepackage{todonotes}

\makeatletter \def\arcr{\@arraycr}
\makeatother

\usepackage{mathpartir}
\usepackage{bbm}
\usepackage{tikz}
\usetikzlibrary{automata, positioning, arrows}
\usepackage{iris}
\usepackage{pftools}
\usepackage{macro}
\usepackage{iblock}

\begin{document}
\title{Context-Dependent Effects and Concurrency in Guarded Interaction Trees}

\author{Sergei Stepanenko}
\email{serst@itu.dk}
\orcid{0000-0002-7322-5644}
\affiliation{\institution{IT University of Copenhagen}
  \city{Copenhagen}
  \country{Denmark}
}
\author{Emma Nardino}
\email{emma.nardino@ens-lyon.fr}
\orcid{0009-0004-2942-8761}
\affiliation{\institution{Univ Lyon, ENS de Lyon, UCBL, CNRS, Inria, LIP, UMR 5668, F-69342, Lyon CEDEX 07}
  \city{Lyon} 
  \country{France}
}
\author{Virgil Marionneau}
\email{virgil.marionneau@ens-rennes.fr}
\orcid{0009-0005-9568-4592}
\affiliation{\institution{ENS Rennes}
  \city{Bruz}
  \country{France}
}
\author{Dan Frumin}
\email{d.frumin@rug.nl}
\orcid{0000-0001-5864-7278}
\affiliation{\institution{University of Groningen}
  \city{Groningen}
  \country{The Netherlands}
}
\author{Amin Timany}
\email{timany@cs.au.dk}
\orcid{0000-0002-2237-851X}
\affiliation{\institution{Aarhus University}
  \city{Aarhus}
  \country{Denmark}
}
\author{Lars Birkedal}
\email{birkedal@cs.au.dk}
\orcid{0000-0003-1320-0098}
\affiliation{\institution{Aarhus University}
  \city{Aarhus}
  \country{Denmark}
}

\renewcommand{\shortauthors}{S. Stepanenko et al.}
\renewcommand{\app}{\mathbin{+\mkern-8mu+}}

\begin{abstract}
  Guarded Interaction Trees are a structure and a fully formalized framework for
  representing higher-order computations with higher-order effects in Rocq.
We present an extension of Guarded Interaction Trees to support formal reasoning about context-dependent effects.
  That is, effects whose behaviors depend on the evaluation context, \eg{}, \texttt{call/cc}, \texttt{shift} and \texttt{reset}.
  Using and reasoning about such effects is challenging since certain compositionality principles no longer hold in the presence of such effects.
  For example, the so-called ``bind rule'' in modern program logics (which allows one to reason modularly about a term inside an evaluation context) is no longer valid.
  The goal of our extension is to support representation and reasoning about context-dependent effects in the most painless way possible.
  To that end, our extension is conservative: the reasoning principles (and the Rocq implementation) for context-independent effects remain the same.
We show that our implementation of context-dependent effects is viable and powerful.
  We use it to give direct-style denotational semantics for higher-order programming languages with \texttt{call/cc} and with delimited continuations.
We extend the program logic for Guarded Interaction Trees to account for context-dependent effects, and we use the program logic to prove that the
  denotational semantics is adequate with respect to the operational semantics.
  This is achieved by constructing logical relations between syntax and semantics inside the program logic.
Additionally, we retain the ability to combine multiple effects in a modular way, which we demonstrate
  by showing type soundness for safe interoperability of a programming language with delimited continuations and a programming language with higher-order store.
Furthermore, as another contribution, in addition to context-dependent effects, we show how to extend Guarded Interaction Trees with preemptive concurrency.
  To support implementation and verification of concurrent data structures and algorithms in the presence of preemptive concurrency one requires atomic state modification operations, \eg{}, compare-and-exchange.
  We demonstrate how the ability to pass lambda functions as effect inputs allows us to define a single generic atomic state modification effect from which various atomic operations can be derived.

\end{abstract}


\maketitle
\section{Introduction}
\label{global:sec:intro}
Despite a lot of recent progress, representing and reasoning about programming languages in proof assistants, such as Rocq, is still considered a major challenge.
The design space is wide and many approaches have been considered.
Recently, research on a novel point in the design space was initiated with the introduction of Interaction Trees~\cite{DBLP:journals/pacmpl/XiaZHHMPZ20}, or ITrees for short.
ITrees were introduced to simplify representation and reasoning about possibly non-terminating programs with side effects in Rocq.
In a sense, ITrees provide a target for denotational semantics of
programming languages, which allows one to abstract from syntactic
details often found in models based on operational semantics.
ITrees specifically allow one to easily represent and reason about various effects and their combinations in a modular way.
A wide range of subsequent applications of ITrees (see, \eg{}, \citet{KohEtAl:2019,ZhangEtAl:2021,ZakowskiEtAl:2021,SilverEtAl:2023}, among others) show that they indeed work very well for representing and reasoning about first-order programs with first-order effects.
As part of the trade-offs, ITrees could not support higher-order representations and higher-order effects.
To address this challenge, Guarded Interaction Trees (or \gitrees for short) were introduced \cite{DBLP:journals/pacmpl/FruminTB24}.

While \gitrees support \emph{higher-order effects}, \eg{}, higher-order store which is known to be particularly challenging, they are limited to
effects that do not alter the control flow of the program, \ie{},
effects that do no modify \emph{the continuation}. As a consequence, one cannot use \gitrees to
give direct-style denotational semantics of programming languages with
context-dependent effects such as \texttt{call/cc}, exceptions, or delimited
continuations. It is of course possible to give semantics to, \eg{},
\texttt{call/cc} using a CPS translation, but this would require a global
transformation which complicates representation and reasoning, especially in
combination with other effects present. In this paper we extend \gitrees to support direct-style representation and reasoning about
higher-order programs with higher-order context-dependent effects in Rocq, and
evaluate its modularity.

We want to stress that our extension to context-dependent effects is not only
theoretically interesting, but also important for scalability, since many real
mainstream programming languages include context-dependent effects.
Indeed, exceptions are now a standard feature in many languages, and while other
context-dependent effects such as delimited continuations are not as widespread in
mainstream programming languages, they are present in the core calculi used for
some such languages: for example, the Glasgow Haskell Compiler core language was
recently extended with delimited continuations to support the introduction of
effect systems, which can be efficiently developed on top of delimited
continuations~\cite{GHC/Cont}. Moreover, effect handlers, which rely on control
flow operators, have recently been introduced in OCaml 5.0, and as a design
feature in Helium~\cite{10.1145/3290319}, Koka~\cite{DBLP:conf/icfp/Leijen17,DBLP:journals/corr/Leijen14}, and other languages.
Note that these languages do not only include forms of delimited continuations, but also other effects, which underscores the importance of considering delimited control in combination with other effects.

\paragraph{Overview of the technical development and key challenges.}
Similarly to how ITrees are defined as a coinductive type in Rocq,
\gitrees are defined as a guarded recursive type.
(This is to support function spaces in GITrees, because in the presence of function
spaces there are negative occurrences of the recursive type and hence one cannot
simply define the type of GITrees using coinduction \citep{DBLP:journals/pacmpl/FruminTB24}).
Rocq does not directly support guarded recursive types, so GITrees are
defined using a fragment of guarded type theory implemented in Rocq as part of the Iris framework \cite{irisWWW}.
To work efficiently with GITrees we make use of other Iris features like separation logic and Iris Proof Mode \cite{DBLP:conf/popl/KrebbersTB17}, which we use to define custom program logics for different (combinations of) effects.
This enables us to reason about \gitrees smoothly in the Iris logic in Rocq, in much the same way as one works directly in Rocq.
We recall the precise definition of GITrees in \Cref{global:sec:gitrees}.

Broadly speaking, \gitrees model effects in the following way.
The type of GITrees is parameterized over a set of effectful operations.
Each operation is given meaning by a \emph{reifier} function, using a form of state monad.
From this, we define the \emph{reduction} relation of GITrees, which gives semantics to computations represented by GITrees.

To support context-dependent effects, we extend (see \Cref{global:sec:callcc-lang}) the notion of a reifier so that reification of effects can also depend on the context; technically, the reifier operation becomes parameterized by a suitable GITree continuation.
This extension allows us to give semantics to context-dependent effects, but it comes at a price.
In particular, following the change in semantics, we need to reformulate the program logic for GITrees:
in the presence of context-dependent effects (like \texttt{call/cc}), the so-called ``bind''-rule becomes unsound.
Of course, we do not want this reformulation to complicate reasoning about computations that do \emph{not} include context-dependent effects.
To that end, we parameterize the GITrees (and the program logic) by a flag, which allows us to recover the original proof rules and make sure that all of the original GITrees framework still works with our extension.

To motivate the extension to context-dependent effects, we give direct-style denotational semantics to a higher-order programming language, $\iocclang$ with \texttt{call/cc} (see \Cref{global:sec:callcc-lang} for details).
Furthermore, we use the derived program logic to construct a logical relation between the denotational and the operational semantics to prove computational adequacy of our model.
To provide guidance on modeling context-dependent effects, we present exceptions as a pedagogical example (\Cref{global:exceptions}), showing step-by-step how to design effects, maintain appropriate state (a handler stack), and derive program logic rules.

Our main interest, however, lies in the treatment of delimited continuations.
In \Cref{global:sec:delim-lang} we show how to represent delimited continuations as effects in \gitrees,
and we use them to define a novel denotational semantics for a
programming language with \texttt{shift} and \texttt{reset} operators.
We prove that our denotational semantics is sound with respect to the operational semantics (given by an extension of the CEK abstract machine \citep{DBLP:conf/ifip2/FelleisenF87}).
We additionally use the program logic to define a logical relation, and prove computational adequacy and \emph{semantic} type soundness.
We recall that semantic type soundness is interesting because it allows one to combine syntactically well-formed programs with syntactically ill-typed, but semantically well-behaved programs \cite{logical-approach-to-type-soundness}.

As we mentioned, it is important to consider delimited continuations not only on their own, but in combinations with other effects.
And indeed, one of the key points of ITrees, and therefore also of \gitrees, is that they support reasoning about effecting and language interoperability by establishing a common unifying semantic framework.
In this paper, we consider (in \Cref{global:sec:ffi}) an example of such interaction: we show a type-safe embedding of $\delimlang$, a language with delimited continuations into a language $\embedstlang$, a language with with higher-order store.
We allow $\delimlang$ expressions to be embedded into $\embedstlang$ by surrounding them by simple glue code, and use a type system to ensure type safety of the combined language.
To define the semantics of the combined language we rely on the modularity of GITrees, and combine reifers for delimited continuations with reifiers for higher-order store.
We prove type safety of the combined language by constructing a logical relation and use the program logic both to define the logical relation and to verify the glue code between the two languages.
The type system for the combined language naturally requires that the embedded code is well-typed according to the type system for $\delimlang$ and thus we can rely on the type soundness of $\delimlang$ (proved in the earlier Section \ref{global:sec:delim-lang}) when proving type safety for the combined language. At the end of \Cref{global:sec:ffi}, we give an example of how to verify a more involved interaction of effects, albeit without the type system.

Additionally, we demonstrate the flexibility of \gitrees by presenting another extension of \gitrees, preemptive concurrency (\Cref{global:concurrency}), which is orthogonal to the extension with context-dependent effects.
The key insight is that \gitrees' support for higher-order effects, specifically, effects that take lambda functions as inputs, enables a generic atomic state modification effect.
This is essential for correctly implementing shared mutable state in a concurrent setting, as we demonstrate by presenting denotational semantics for a language featuring fork and atomic exchange operations.

\paragraph{Summary of Contributions.}
In summary, we present:
\begin{itemize}
\item A conservative (with respect to the old results) extension to
  \gitrees for representing and reasoning about context-dependent effects (\Cref{global:sec:callcc-lang}).
\item A sound and adequate model of a calculus with \texttt{call/cc} and \texttt{throw},
  implemented in a direct style (\Cref{global:sec:callcc-sem}).
\item A tutorial-style presentation of modeling exceptions, providing step-by-step guidance on designing context-dependent effects in \gitrees (\Cref{global:exceptions}).
\item A sound and adequate model of a calculus with delimited
  continuations, with operations \texttt{shift} and \texttt{reset}, implemented in a direct style (\Cref{{global:sec:delim-lang}}).
\item A type system for interoperability between a programming language with delimited continuations and a programming language with higher-order store, with a semantic type safety proof (\Cref{{global:sec:ffi}}).
\item An extension of \gitrees with preemptive concurrency, demonstrating how lambda functions in effects enable generic atomic operations, with a sound model of an affine language with fork and and runtime linearity checks using atomic exchange (\Cref{global:concurrency}).
\end{itemize}
All results in the paper have been formalized in Rocq as a modification to the GITrees library and the previously proved results have been ported to our extension.
We conclude and discuss related work in \Cref{global:sec:conc}.
Before we go on with the main part of the paper, we recall some background material on \gitrees.

\paragraph{Contributions Compared to the ESOP'25 \citep{DBLP:conf/esop/StepanenkoNFTB25} Version of this Paper.}
Most of this publication is an almost verbatim repetition of the earlier publication \citep{DBLP:conf/esop/StepanenkoNFTB25} with minor adjustment and updates, \eg{}, fixing typos, making the text, including this \nameref{global:sec:intro} section, consistent with the additions enumerated below.
Compared to this earlier version \citep{DBLP:conf/esop/StepanenkoNFTB25} the present publication contains two additional sections: 
\begin{itemize}
  \item \Cref{global:exceptions} contains an example of providing semantics for a language with exceptions
  \item \Cref{global:concurrency} contains details of extending Guarded Interaction Trees with support for preemptive concurrency with atomic effects
\end{itemize}
\Cref{global:sec:intro,global:sec:gitrees,global:sec:callcc-lang,global:sec:delim-lang,global:sec:ffi,global:sec:conc} appear in the first author's PhD dissertation.

\section{Guarded Interaction Trees}
\label{global:sec:gitrees}

In this section we provide an introduction to guarded interaction trees.
Our treatment is brief, and we refer the reader to the original paper for details~\cite{DBLP:journals/pacmpl/FruminTB24}.

\paragraph{Iris and Guarded Type Theory.}
Guarded Interaction trees (\gitrees{}) are defined in Iris logic. Here
we briefly touch Iris, and refer the reader to the literature on
Iris~\cite{DBLP:journals/jfp/JungKJBBD18} and guarded type
theory~\cite{DBLP:journals/corr/abs-1208-3596} for more in-depth
details.
Iris is a separation logic framework built on top of a model of \emph{guarded type theory}, the main use of which is to solve recursive equations and define guarded recursive types, such as the type of \gitrees described below.
Moreover, Iris has a specialized proof mode~\cite{DBLP:conf/popl/KrebbersTB17}, implemented in Rocq.
This allows the users of Iris to carry out formal reasoning in separation logic as if they are proving things normally Rocq, as we have done in the formalization of this work.
For this reason, in the paper we will work with Iris and its type theory informally.
Still, we need to say a few things about the foundations.

The syntax of Iris, shown in~\Cref{fig:gt:grammar}, contains types, terms, and propositions.
The grammar is standard for higher-order logic, with the exception of the guarded types fragment, and separation logic connectives.
The type of propositions is denoted $\Prop$.
The \emph{guarded} part of guarded type theory is the ``later'' modality $\latert$.
Intuitively, we view all types as indexed by a natural number, where $\tau_{n}$ contains elements of $\tau$ ``at time'' $n$.
Then $\latert \tau$ contains elements of $\tau$ at a later time; that is, $(\latert \tau)_{n}=\tau_{n-1}$.
There is an embedding $\Next : \tau \to \latert \tau$, and there is a \emph{guarded} fixed point combinator $\fix_{\tau}: (\latert \tau \to \tau) \to \tau$, similar to the unguarded version in PCF.
We can also lift functions to $\latert$: given $f : A \to B$, we have $\latert f : \latert A \to \latert B$.

For the proposition, Iris contains the usual separation logic connectives, and the two modalities: ``later'' $\later$ and ``persistently'' $\Box$.
The propositional $\later$ modality reflects the type-level later modality $\latert$ on the level of propositions, as justified by the following rule: $\later (\alpha =_\tau \beta) \dashv \vdash \Next(\alpha) =_{\latert \tau} \Next(\beta)$.
The persistence modality $\Box P$ states that the proposition $P$ is available without claiming any resources (as it normally is the case in separation logic); crucially it makes the proposition duplicable: $\Box P \vdash (\Box P)\ast (\Box P)$. An example of a persistent proposition is the
invariant proposition $\knowInv{}{P}$, which satisfies $\knowInv{}{P} \vdash \Box \knowInv{}{P}$.

\paragraph{Guarded Interaction Trees.}
\begin{figure}[t]
  \begin{align*}
    \type \bnfdef{}&
                     \Prop \mid
                     0 \mid
                     \Tunit \mid
                     \Tbool \mid \Tnat \mid
                     \type + \type \mid
                     \type \times \type \mid
                     \type \to \type \mid \latert \type \mid
                     I \mid
                     \Sigma_{\idx\in \mathtt{I}} \type_\idx \mid
                     \Pi_{\idx\in \mathtt{I}} \type_\idx \mid
                     \dots
    \\[0.4em]
    \term \bnfdef{}&
                     \var \mid
                     \sigfn(\term_1, \dots, \term_n) \mid
                     \textlog{abort}\; \term \mid
                     () \mid
                     (\term, \term) \mid
                     \pi_i\; \term \mid
                     \Lam \var:\type.\term \mid \\
    & \textlog{inj}_i\; \term \mid
    \textlog{match}\; \term \;\textlog{with}\; \overline{\textlog{inj}_i\; \var. \term} \;\textlog{end} \mid
    \Next(\term) \mid
    \fix_\type \mid \dots
    \\[0.4em]
    \prop \bnfdef{}&
                     \FALSE \mid
\term =_\type \term \mid
\prop \lor \prop \mid
                     \prop \to \prop \mid
\All \var:\type. \prop \mid
    \prop \ast \prop \mid
    \prop \wand \prop \mid
\Box \prop \mid
    \knowInv{}{P} \mid
\later \prop \mid \dots
\end{align*}
  \caption{Grammar for the Iris base logic.}
  \label{fig:gt:grammar}
\end{figure}
\begin{figure}[t]
\begin{align*}
    & \mathsf{guarded\ type}\ \IT_E(A) {}= \Rret : A \to \IT_E(A) \\
    &\ \ALT \Fun : \latert (\IT_E(A) \to \IT_E(A)) \to \IT_E(A)\\
    &\ \ALT \Err : \Error \to \IT_E(A)\\
    &\ \ALT \Tau : \latert \IT_E(A) \to \IT_E(A) \\
    &\ \ALT \Vis : \prod_{\idx \in \mathtt{I}} \big( \Ins_{\idx}(\IT_E(A)) \times (\Outs_{\idx}(\IT_E(A)) \to \latert \IT_E(A))\big) \to \IT_E(A)
  \end{align*}
  \caption{Guarded datatype of interaction trees.}
  \label{fig:gt:gitrees_def}
\end{figure}
Guarded recursive datatypes are datatypes obtained from recursive equations of the form $X = F(\latert X)$.
In other words, guarded recursive datatypes are similar to the regular datatypes you see in normal programming languages, but every recursive occurrence of the type must be guarded by the $\latert$ modality.
The datatype we are concerned with here is the type of \gitrees, shown in \Cref{fig:gt:gitrees_def}.
It is parameterized over two types: the ground type $A$ and the effect signature $E$ (more on it below).

Guarded Interaction Trees represent computational trees in which the leaves are of the ground type ($\Rret(a)$), error states ($\Err(e)$), and functions ($\Fun(f)$).
The leaves $\Rret(a)$ and $\Fun(f)$ are also called \emph{values}, and we write $\ITv_{E}(A)$ for the type of $\IT_{E}(A)$-values.

The nodes of the computation trees are of the two kinds.
The first one is a ``silent step'' constructor $\Tau(\alpha)$.
It represents an unobservable internal step of the computation.
For convenience, we use the function $\Tick \eqdef \Tau \circ \Next \col \IT_E(A) \to \IT_E(A)$
that ``delays'' its argument.
This function satisfies the following equation: $\Tick(\alpha) = \Tick(\beta) \dashv \vdash \later (\alpha = \beta)$.

The second kind of nodes are effects given by $\Vis_{i}(x,k)$.
The parameters $I$, $\Ins$ and $\Outs$ are part of the effect signature $E$.
The set $I$ is the set of \emph{names} of operations. The \emph{arities} of an operation $i\in I$ are given by functors $\Ins_{i}$ and $\Outs_{i}$.
Let us give an example.

Consider the following signature for store effects.
The signature $E_{\mathsf{state}}$ consists of effects $\{\mathtt{write}, \mathtt{read}, \mathtt{alloc}\}$ with the following input/output arities:
\begin{align*}
  \Ins_{\mathtt{write}}(X) &\eqdef \Loc \times \latert X & \Ins_{\mathtt{read}}(X) &\eqdef \Loc & \Ins_{\mathtt{alloc}}(X) & \eqdef \latert X \\
  \Outs_{\mathtt{write}}(X) &\eqdef \Tunit & \Outs_{\mathtt{read}}(X) &\eqdef \latert X & \Outs_{\mathtt{alloc}}(X) &\eqdef \Loc
\end{align*}
For example, $\mathtt{write}$ expect a location and a new \gitree as its input, and simply returns the unit value as an output.
We write $\Vis_{\mathtt{write}}((\ell, \alpha), \Lam \_. \beta)$ for the computation that invokes the $\mathtt{write}$ effect with arguments $\ell$ and $\alpha$, waits for it to return, and proceeds as $\beta$.
Thus, the first argument for $\Vis_{i}$ is the input, and the second one is the continuation dependent on the output. This continuation determines the branching in (G)Itrees.

For effects like above, it is usually convenient to provide wrappers:
\begin{align*}
  \ALLOC(\alpha \col \IT, k \col \Loc \to \IT) &\eqdef \Vis_{\mathtt{alloc}}(\Next(\alpha), \Next \circ k) \\
  \READ(\ell \col \Loc) &\eqdef \Vis_{\mathtt{read}}(\ell, \lambda x. x) \\
  \WRITE(\ell \col \Loc, \alpha \col \IT) &\eqdef \Vis_{\mathtt{write}}((\ell, \Next(\alpha)), \lambda x . \Next(\Rret(\inj())))
\end{align*}

When the signature and the return type are clear from the context, we simply write $\IT$ and $\ITv$ for the \gitrees and \gitree-values.
\paragraph{Equational theory.}
\gitrees come with a number of operations (defined using the recursion principle) that are used for writing and composing computations.
Here we list some of those operations which we will be using.
\begin{figure*}[t]
\begin{align*}
      \getval(\Rret(a),f) &= f(\Rret(a)) & \getval(\Tau(t),f) &= \Tau(\latert \getval(t,f)) \\
      \getval(\Fun(g),f) &= f(\Fun(g)) & \getval(\Tick(\alpha),f) &= \Tick(\getval(\alpha,f)) \\
      \getval(\Err(e),f) &= \Err(e) & \getval(\Vis_i(x,k),f) &= \Vis_i(x, \latert \getval(-,f) \circ k)
    \end{align*}
\caption{Example function on Guarded Interaction Trees.}
  \label{fig:gt:functions}
\end{figure*}
The function $\getval(\alpha, f : \ITv \to \IT)$ are used for sequencing computations, and its corresponding equations are shown in \Cref{fig:gt:functions}.
Intuitively, $\getval(\alpha, f)$ first tries to compute $\alpha$ to a value (a $\Rret(a)$ or a $\Fun(g)$), and then calls $f$ on that value.
Similarly, $\getfun(\alpha: \IT, f : \latert(\IT \to \IT) \to \IT)$ and $\getret(\alpha \col \IT_E(A), f : A \to \IT_E(A))$ first compute $\alpha$ to a value; if that value is a function $\Fun(g)$ (resp., $\Rret(a)$), then it proceeds with $f(g)$ (resp., $f(a)$).
Otherwise it results in an runtime error.

Crucially, to work with higher-order computations, \gitrees provide the ``call-by-value'' application $\APPs{\alpha}{\beta}$ satisfying the following equations:
\begin{align*}
  \APPs{\alpha}{\Tick(\beta)} &= \Tick(\APPs{\alpha}{\beta})
  & \APPs{\alpha}{\Vis_i(x,k)} &= \Vis_i(x, \Lam y. \APPsl{\Next(\alpha)}{k\ y})\\
  \APPs{\Tick(\alpha)}{\beta_v} &= \Tick(\APPs{\alpha}{\beta_v})
  & \APPs{\Vis_i(x,k)}{\beta_v} &= \Vis_i(x, \Lam y. \APPsl{k\ y}{\Next(\beta_v)})\\
  \APPs{\Fun(\Next(g))}{\beta_v} &= \Tick(g(\beta_v))
  & \APPs{\alpha}{\beta} &= \Err(\RunTime) \mbox{ in other cases }
\end{align*}
where $\APPsl{-}{-}$ is defined as the lifting of $\APPs{-}{-}$ to $\latert \IT_{E}(A) \to \latert \IT_{E}(A) \to \latert \IT_{E}(A)$, and $\beta_v \in \ITv_{E}(A)$ is either $\Rret(a)$ or $\Fun(g)$.

The application function $\APPs{\alpha}{\beta}$ simulates strict function application.
It first tries to evaluate $\beta$ to a value $\beta_{v}$.
Then it tires to evaluate $\alpha$ to a function $f$.
If it succeeds, then it invokes $f(\beta_{v})$.
If at any point it fails, application results in a runtime error.

For the often-used case of \gitrees where the ground type includes natural numbers, we use the function $\NATOP : (\mathbb{N} \to \mathbb{N} \to \mathbb{N}) \to IT \to IT \to IT$ which lifts binary functions on natural numbers to binary functions on \gitrees.
That is, $\NATOP_{f}(\alpha, \beta)$ first evaluate \gitrees $\beta$ and $\alpha$ to values.
If those values are natural numbers, then it computes $f$ of those numbers and returns the result as a \gitree.
Otherwise,  it returns a runtime error $\Err(\RunTime)$.

\paragraph{Reification and reduction relation.}
The semantics for effects are given in terms of \emph{reifiers}.
A reifier for the signature $E$ is a tuple $(\State, r)$, where
$\State$ is a type representing the internal state needed to reify the effects, and $r$ is a reifier function of the type given in \Cref{fig:gt:reifier_sig}.
The idea is that $r_{i}$ uses the internal state $\State$ to compute the output of the effect $i$ based on its input.
\begin{figure}[t]
  \[
    r : \prod_{\idx \in E} \Ins_{\idx}(\IT_E) \times \stateO \to \optionO(\Outs_{\idx}(\IT_E) \times \stateO)
  \]
  \begin{mathpar}
    \infer
    {r_i(x,\sigma) = \Some(y, \sigma') \and k\ y = \Next(\beta)}
    {\reify(\Vis_i(x, k), \sigma) = (\Tick(\beta), \sigma')}
    \and
    \infer
    {r_i(x,\sigma) = \None}
    {\reify(\Vis_i(x, k), \sigma) = (\Err(\RunTime), \sigma)}
  \end{mathpar}
  \caption{Signature of reifiers and the reification function}
  \label{fig:gt:reifier_sig}
\end{figure}

For example, for the store effects we take $\State$ to be a map from locations to $\latert \IT$ (representing the heap); and we define the following reifier functions:
\begin{align*}
  r_{\mathtt{write}}((\ell, \alpha),\sigma) &= \Some((), \sigma[\ell \mapsto \alpha]) \tag{where $\ell \in \sigma$, and $\None$ otherwise} \\
  r_{\mathtt{read}}(\ell, \sigma) &= \Some(\alpha, \sigma) \tag{where $\sigma(\ell) = \alpha$, and $\None$ otherwise} \\
  r_{\mathtt{alloc}}(\alpha, \sigma) &= \Some(\ell, \sigma[\ell \mapsto \alpha]) \tag{where $\ell \notin \sigma$}
\end{align*}

Given reifiers for all the effects, we
define a function $\reify : \IT \times \stateO \to \IT \times \stateO$ (as in \Cref{fig:gt:reifier_sig})
that, given $(\alpha, \sigma)$ reifies the top-level effect in $\alpha$ using the state $\sigma$, and returns the reified \gitree and the updated state.

The $\reify$ function is then used to give reduction semantics for \gitrees{}.
We write $(\alpha,\sigma) \istep (\beta, \sigma')$ for such a reduction step.
The definition of $\istep$ is given internally in the logic:
\begin{align*}
  (\alpha,\sigma) \istep (\beta,\sigma') &\eqdef \big(\alpha = \Tick(\beta) \wedge \sigma = \sigma' \big) \\ &\vee \big(\Exists i\,x\,k. \alpha = \Vis_i(x,k)
   \wedge \reify(\alpha,\sigma) = (\Tick(\beta),\sigma')\big)
\end{align*}
That is, either $\alpha$ is a ``delayed'' computation $\Tick(\beta)$, which then reduces to $\beta$; or it is an effect that can be reified.
Recall that we write $\Tick$ for the composition $\Tau \circ \Next$.

Note that the $\reify$ function operates on the top-level effect of the \gitree.
But what if the top-level constructor is not $\Vis$, \eg{} if we have an effect inside an ``evaluation context''?
The role of evaluation contexts in \gitrees is played by \emph{homomorphisms}, which also allow us to bubble up necessary effects to the top of the \gitree.
\begin{definition}[Homomorphism]
  \label{def:hom}
  A map $f : \IT \to \IT$ is a homomorphism, written $f \in \Hom$, if it satisfies:
  \begin{align*}
  f(\Err(e)) = \Err(e) \quad
  f(\Tick(\alpha)) = \Tick(f(\alpha)) \quad
  f(\Vis_i(x,k)) = \Vis_i(x, \latert f \circ k)
  \end{align*}
\end{definition}
For example, $\Lam x. \APPs{\alpha}{x}$ is a homomorphism, and so is $\Lam x. \getval(x, f)$. On the other hand, $\Lam x. \Vis_{\mathtt{alloc}}(\Next(x), k)$ (for some fixed $k$) is \emph{not} a homomorphism.

\paragraph{Program logic.}
In order to reason about \gitrees, we employ the full power of the Iris separation logic framework.
The program logic operates on the propositions of the form $\wpre{\alpha}{\Phi}$.
This weakest precondition proposition intuitively states that the \gitree $\alpha$ is safe to reduce, and when it fully reduces, the resulting value satisfies the predicate $\Phi$.
Another important predicate is $\hasstate(\sigma)$, which signifies ownership of the current state $\sigma$.
\begin{figure}[t]
  \begin{mathpar}
    \inferH{wp-reify}
    {
      \begin{array}[c]{c}
      \hasstate(\sigma) \arcr
      \reify(\Vis_i(x,k), \sigma) = (\Tick(\beta), \sigma')
      \arcr
        \later\big(\hasstate(\sigma') \wand \wpre{\beta}{\Phi} \big)
      \end{array}
        }
    {\wpre{\Vis_i(x, k)}{\Phi}}
    \and
    \inferH{wp-hom}
    {f \in \Hom \quad \wpre{\alpha}{\Ret \beta_v. \wpre{f(\beta_v)}{\Phi}}}
    {\wpre{f(\alpha)}{\Phi}}
\end{mathpar}
  \caption{Selected weakest precondition rules.}
  \label{fig:gt:wp_rules}
\end{figure}

In \Cref{fig:gt:wp_rules} we show the rules, on which we focus in this work. Let us describe their meaning. The rule \ruleref{wp-reify} allows us to symbolically execute effects in \gitrees{}.
It is given in a general form, and is used to derive domain-specific rules for concrete effects.
Another important rule is \ruleref{wp-hom} which allows one to separate the reasoning about the computation from the reasoning about the context.
The reason why \ruleref{wp-hom} is sound (this is going to be important in the next section when we make it unsound), is because the reduction $\istep$ of \gitrees satisfies the following properties which allow one disentangle a homomorphism from the \gitree it's applied to:
\begin{lemma}
  \label{lem:r:hom_istep}
  Let $f$ be a homomorphism.
  Then,
  \begin{itemize}
  \item $(\alpha,\sigma)\istep(\beta,\sigma')$ implies
    $(f(\alpha),\sigma)\istep(f(\beta),\sigma')$;
  \item   If $(f(\alpha), \sigma)\istep(\beta',\sigma')$ then either
    $\alpha$ is a \gitree-value, or
    there exists $\beta$ such that
    $(\alpha,\sigma)\istep(\beta,\sigma')$ and $\later(f(\beta) = \beta')$.
  \end{itemize}
\end{lemma}

Finally, as usual in Iris, the program logic satisfies an adequacy property, which allows one to relate propositions proved in the logic to the actual semantics:
\begin{theorem}
  \label{thm:wp_adequacy}
  Let $\alpha$ be an interaction tree and
  $\sigma$ be a state such that
  \[\hasstate(\sigma) \vdash \wpre{\alpha}{\Phi}
  \] is derivable for some meta-level predicate $\Phi$ (containing
  only intuitionistic logic connectives). Then for any $\beta$ and
  $\sigma'$ such that $(\alpha, \sigma) \istep^{\ast} (\beta,\sigma')$,
  one of the following two things hold:
  \begin{itemize}
  \item \emph{(adequacy)} either $\beta \in \ITv$, and $\Phi(\beta)$
    holds in the meta-logic;
  \item \emph{(safety)} or there are $\beta_1$ and $\sigma_1$ such
    that $(\beta,\sigma')\istep(\beta_1,\sigma_1)$
  \end{itemize} In particular, safety implies that $\beta \neq
  \Err(e)$ for any error $e \in \Error$.
\end{theorem}
The role of meta-logic is played by the Rocq system; thus, the adequacy theorem allows us to relate proofs inside the program logic (Iris) to the proofs on the level of Rocq.
This aspect is important in Iris and \gitrees in general, but it is orthogonal to the work that we present in this paper.
See \cite{DBLP:journals/pacmpl/FruminTB24} for more details.

\section{Context-Dependent Reification}
\label{global:sec:callcc-lang}
In this section we extend reification to handle context-dependent effects, using a language $\iocclang$ with \texttt{call/cc} as a concrete example.
In \Cref{sec:opsem:iocclang} we present $\iocclang$'s syntax and operational semantics (in the usual style with evaluation contexts).
We then show why the current \gitrees framework cannot be used as a denotational model for $\iocclang$ directly.
In \Cref{global:sec:str-reifiers} we introduce our generalization of reification for context-dependent effects and corresponding extensions to the \gitrees program logic.
In \Cref{global:sec:callcc-sem} we demonstrate that our extension works as intended: we give a denotational semantics for $\iocclang$, and we show how the general program logic for \gitrees specializes to a logic for reasoning about \texttt{call/cc}.
We prove soundness and computational adequacy of denotational semantics using a logical relation defined within our program logic.

\subsection{Operational Semantics and Type System for $\iocclang$}
\label{sec:opsem:iocclang}
By $\iocclang$ we denote a simply-typed $\lambda$-calculus with natural numbers, recursive functions, and \texttt{call/cc}.
The relevant pieces of syntax, the type system and the operational semantics are given in \Cref{fig:cc:lang}. The type system includes natural numbers, function types, and the type $\tcont{\tau}$ of continuations.
The $\callcc{\var}{\expr}$ expression takes the current evaluation context and binds it to $\var$ in $\expr$.
The  $\throw{\expr}{\expr'}$ expression evaluates passes the first argument (value) to the second argument (continuation, represented as an evaluation context).

The operational semantics of $\iocclang$ is separated into two layers.
The first layer consists of local reductions of primitive expressions ($\contrA{e}{e'}{K}$), and the second layer
lifts local reductions to reductions among complete programs ($\contr{\plug{K}{\expr}}{\plug{K'}{e'}}$).
The local reductions are parameterized by an
evaluation context $K$, which allows $\callcc{\var}{\expr}$ to capture the evaluation context.
\begin{figure}[t]
    \begin{grammar}
    \text{types} & \Tys \ni \tau & \tnat \mid \tarr{\tau_1}{\tau_2} \mid \tcont{\tau} \\
    \text{expressions} & \Expr \ni \expr & \val \mid \var \in \Var \mid \eapp{\expr_1}{\expr_2} \mid \natop{\expr_1}{\expr_2}
\mid \If \expr_1 then \expr_2 \Else \expr_3 \\
    \GrmContinue & \mid \callcc{\var}{\expr} \mid \throw{\expr_1}{\expr_2} \\
    \text{values} & \Val \ni \val & n \mid \Rec f \var = \expr \mid \cont{K} \\
    \text{eval. cont.} & \Ectx \ni K & \emptyK \mid \If K then \expr_1 \Else \expr_2 \mid \eapp{K}{\val}
\mid \eapp{\expr}{K} \mid \natop{\expr}{K} \mid \natop{K}{\val} \\
    \GrmContinue & \mid \throw{K}{\expr} \mid \throw{\val}{K}
  \end{grammar}
  \vspace{-0.6cm}
  \begin{mathpar}
\inferrule{}{\contrA{\callcc{\var}{\expr}}{\subst{\expr}{\var}{\cont{K}}}{K}}
    \and
    \inferrule{}{\contr{\plug{K}{\throw{\val}{\cont{K'}}}}{\plug{K'}{\val}}}
    \and
    \inferrule{\contrA{\expr_1}{\expr_2}{K}}{\contr{\plug{K}{\expr_1}}{\plug{K}{\expr_2}}}
  \end{mathpar}
  \vspace{-0.3cm}
    \begin{mathpar}
    \inferrule{\typeA{\Gamma, \var \col \tcont{\tau}}{\expr}{\tau}}{\typeA{\Gamma}{\callcc{\var}{\expr}}{\tau}}
    \and
    \inferrule{\typeA{\Gamma}{\expr_1}{\tau} \\ \typeA{\Gamma}{\expr_2}{\tcont{\tau}}}{\typeA{\Gamma}{\throw{\expr_1}{\expr_2}}{\tau'}}
  \end{mathpar}
  \caption{Syntax and fragments of type system and operational semantics of $\iocclang$.}
  \label{fig:cc:lang}
\end{figure}
Let us now consider what happens if we try to give a direct-style denotational
semantics of $\iocclang$ into GITrees. By direct we mean that we wish to give a
direct interpretation of types and expressions, rather than going through a
global CPS conversion. To define the semantics, we first need to provide an
effect signature, state, and reifiers for each effect, and then we can define the
interpretation of the expressions of the language.

The effect signature, shown in~\Cref{fig:cc:io_constructors}, contains two
effects $\mathtt{callcc}$ and $\mathtt{throw}$.
Since \texttt{call/cc} binds a continuation, it is natural to let the
input arity for $\mathtt{callcc}$ be a callback $(\latert \IT \to \latert \IT) \to \latert \IT$.
The output arity is simply $\latert \IT$.

The input arity for $\mathtt{throw}$ signifies that \texttt{throw}
takes as input an expression and a continuation,
which are represented respectively as $\latert \IT$ and
$\latert (\IT \to \IT)$.
The output arity of $\mathtt{throw}$ is simply the empty type $0$, because
\texttt{throw} never returns.

Note that the input types of $\mathtt{callcc}$ and $\mathtt{throw}$
have slightly different arities. However, we can always transform $f
\col (\latert X \rightarrow \latert X)$ into an element of type
$\latert (X \rightarrow X)$ by performing a silent step in the
function's body: $f' \eqdef \Next(\Lam x. \Tau(f (\Next(x))))$. And we
can always transform $f \col \latert (X \rightarrow X)$ into an
element of type $\latert X \rightarrow \latert X$ using the applicative
structure of the later modality.

For convenience, we will use the abbreviations $\callccGT(f)$ and
$\throwGT(e)$, defined in \Cref{fig:cc:io_constructors}, for
representing denotations of \texttt{throw} and \texttt{call/cc} as
effects in GITrees.
\begin{figure}[t]
\begin{align*}
    \Ins_{\mathtt{callcc}}(X) &\eqdef ((\latert X \rightarrow \latert X) \rightarrow \latert X)
    & \Outs_{\mathtt{callcc}}(X) &\eqdef \latert X
\\
     \Ins_{\mathtt{throw}}(X) &\eqdef \latert X \times \latert (X \rightarrow X)
    & \Outs_{\mathtt{throw}}(X) &\eqdef 0 \\
    \callccGT(f) &\eqdef \Vis_{\mathtt{callcc}}{(f, \idfun)} & \throwGT(e, f) &\eqdef \Vis_{\mathtt{throw}}{(e, f, \Lam x. \elimF x)}
  \end{align*}
\caption{Signatures and opertaions on \gitrees{} with \texttt{call/cc}.}
  \label{fig:cc:io_constructors}
\end{figure}

To complete all the ingredients for the denotational semantics, we need reifiers for the $\mathtt{callcc}$ and $\mathtt{throw}$ effects.
Given our operational understanding of continuations, the natural choice for the local state type $\State$ is $\Tunit$ (since we do not have any state).
However, the current reifier signature (\Cref{fig:gt:reifier_sig}) poses a problem.
Reifiers, as they are now, cannot access their current continuation, which is essential for both effects.
$\callccGT(f)$ needs to pass the current continuation to $f$, while $\throwGT$ must redirect control to a provided continuation instead of returning normally.
The current reifiers lacks this capability, and in the next subsection we show how to generalize the notion of reification to context-dependent effects.

\subsection{Context-dependent Reifiers}
\label{global:sec:str-reifiers}
This section presents our extension to context-dependent reification, and the limitations it imposes on the program logic.
In order to allow reifiers to manage continuations, we change the type of reifiers to accept continuations as an extra parameter, as shown in \Cref{fig:r:different_reifiers}.
Continuations for a given effect are functions from the effect's outputs to \gitrees{}: $\Outs_i(\IT_E) \to \latert \IT_E$.
\begin{figure}[t]
  \begin{align*}
    & r \col \prod_{\idx \in E} \Ins_{\idx}(\IT_E) \times \stateO \times (\Outs_{\idx}(\IT_E) \rightarrow \latert \IT_E) \to \optionO(\latert \IT_E \times \stateO)
  \end{align*}
  \vspace{-0.2cm}
  \begin{mathpar}
    \infer
    {r_i(x,\sigma,\kappa) = \Some(\beta, \sigma')}
    {\reify(\Vis_i(x, \kappa), \sigma) = (\Tau(\beta), \sigma')}
    \and
    \infer
    {r_i(x,\sigma,\kappa) = \None}
    {\reify(\Vis_i(x, \kappa), \sigma) = (\Err(\RunTime), \sigma)}
  \end{mathpar}
  \caption{Type of context-dependent reifiers and the context-dependent reify function}
  \label{fig:r:different_reifiers}
\end{figure}
Given a set of context-dependent reifiers, we define a context-dependent $\mathtt{reify}$ function, also shown in \Cref{fig:r:different_reifiers}.
As before, $\mathtt{reify}$ dispatches to the correct individual reifier for the effect.
Note that now it is the user's responsibility to pass the output of an effect to the given continuation if the control flow is not supposed to be interrupted.
For example, since the evaluation of a $\callcc{\var}{\expr}$ expression does not modify the control flow itself, but simply passes the current continuation to its body, the context-dependent reifier for
$\mathtt{callcc}$ is simply $r_{\mathtt{callcc}}
(x, \sigma, \kappa) = \Some(\kappa~(x~\kappa), \sigma)$.

Before we move on to discussing the consequences of this for program logic, we would like to note that our treatment of continuation (with top-level reifiers dispatching them) parallels Cartwright and Felleisen's  ``extensible direct models''~\cite{cartwright:1994}, which also aimed to support extensible denotational semantics in classical domain theory.
We discuss this more in \Cref{global:sec:conc}.

\paragraph{Program logic for GITrees in the presence of context-dependent reifiers.}
To reflect the generalization to context-dependent reifiers in the program logic,
we replace the proof rule \ruleref{wp-reify} by \ruleref{wp-reify-ctx-dep}, shown in~\Cref{fig:callcc:wp}.
This is, however, not the only change we need to make.
In the presence of context-dependent effects,
\ruleref{wp-hom} is not sound! (A similar observation
was also made by \cite{DBLP:journals/pacmpl/TimanyB19}
in their development of a program logic for \texttt{call/cc}.)
The reason is that context-dependent
reification invalidates~\Cref{lem:r:hom_istep}.
Now, since \ruleref{wp-hom} is not sound anymore, one might expect that we
need to adapt all the other program logic rules to include a homomorphism
similarly to how the rules of \cite{DBLP:journals/pacmpl/TimanyB19}
were adapted to include an evaluation context. However, this is not
necessary, because our program logic is defined on \emph{denotations} on which
we have a non-trivial equational theory, which can be used to reason
about `pure' GITrees. Only for effectful operations, the proof rules will now have
to include a surrounding homomorphism. For instance, \ruleref{wp-write} from \cite{DBLP:journals/pacmpl/FruminTB24} is generalized to \ruleref{wp-write-ctx-dep}, and considers ambient homomorphisms explicitly.

\begin{figure}[t]
  \begin{mathpar}
    \inferH{wp-write}
    {
      \begin{array}[c]{c}
        \heapctx \quad
        \later \ell \mapsto \alpha \arcr
        \later (\ell \mapsto \beta \wand \arcr \quad \wpre{\Rret ()}{\Phi})
      \end{array}
    }
    {\wpre{\WRITE(\ell, \beta)}{\Phi}}
    \and
    \inferH{wp-write-ctx-dep}
    {
      \begin{array}[c]{c}
        \kappa \in \Hom \arcr
        \heapctx \quad
        \later \ell \mapsto \alpha \arcr
        \later (\ell \mapsto \beta \wand \arcr \quad \wpre{\eapp{\kappa}{(\Rret ())}}{\Phi})
      \end{array}
    }
    {\wpre{\kappa\ (\WRITE(\ell, \beta))}{\Phi}}
    \and
    \inferH{wp-reify-ctx-dep}
    {
      \begin{array}[c]{c}
        \hasstate(\sigma) \arcr
        r_i(x, \sigma, k) = \Some(\Next(\beta), \sigma')
        \arcr
        \later\big(\hasstate(\sigma') \wand \arcr \quad \wpre{\beta}{\Phi} \big)
      \end{array}
    }
    {\wpre{\Vis_i(x, k)}{\Phi}}
  \end{mathpar}
  \caption{Program logic in the presence of context-dependent reifiers.}
  \label{fig:callcc:wp}
\end{figure}

Our context-dependent reification extension, though simple, allows us
to build sound and adequate denotational models for languages with
control-flow operators, including $\iocclang$ (shown in the next
subsection). Moreover, our extension is conservative, and we recover
previous case studies (computational adequacy of $\iolangold$ and type
safety for $\afflang$~\cite{DBLP:journals/pacmpl/FruminTB24}) with
minimal modifications; see the accompanying Rocq formalization.

\subsection{Denotational Semantics of $\iocclang$}
\label{global:sec:callcc-sem}
In this section we show that context-dependent reifiers are sufficient for providing a sound and adequate semantic model of $\iocclang$.
We define context-dependent reifiers for $\mathtt{callcc}$ and $\mathtt{throw}$, then prove that this gives a sound interpretation w.r.t.
operational semantics.
To show adequacy, we define a logical relation, which relates the denotational and operational semantics.
The logical relation is defined in the (updated) program logic for \gitrees (following the approach in \cite{DBLP:journals/pacmpl/FruminTB24}), and validates the utility of the program logic.

\begin{figure}[t]
  \begin{align*}
    &\denotE{\var}_\rho = \rho(\var) \\
&\denotE{\callcc{\var}{\expr}}_{\rho} = \callccGT(\Lam (f \col \latert \IT \rightarrow \latert \IT). \denotE{\expr}_{\rho[\var \mapsto \Fun (\Next (\Lam y. \Tau (f (\Next (y)))))]}) \\
    &\denotE{\throw{\expr_1}{\expr_2}}_{\rho} = \getval (\denotE{\expr_1}_\rho, \Lam x. \getfun (\denotE{\expr_2}_\rho, \Lam f. \throwGT(x, f)))\\
    &\denotV{\cont K}_\rho = \Fun(\Next(\Lam x. \Tau(\APPsl{\denotK{K}_\rho}{\Next(x)})))
\\
&\denotK{\throw{K}{\expr}}_\rho = \Lam x. \getval (\denotK{K}_\rho\ x, \Lam y. \getfun (\denotE{\expr}_\rho, \Lam f. \throwGT(y, f)))\\
    &\denotK{\throw{\val}{K}}_\rho = \Lam x. \getval (\denotV{\val}_\rho, \Lam y. \getfun (\denotK{K}_\rho\ x, \Lam f. \throwGT(y, f)))
  \end{align*}
  \caption{Denotational semantics of $\iocclang$ (selected clauses).}
  \label{fig:cc:lang_sem}
\end{figure}

\paragraph{Interpretation of $\iocclang$.}
The denotational semantics of $\iocclang$ is shown in \Cref{fig:cc:lang_sem} (selected clauses only; see Rocq formalization for the complete definition).
The interpretation is split into three parts: $\denotE{-}$ for expressions, $\denotV{-}$ for values, and $\denotK{-}$ for contexts.
For the interpretation of $\throw{\expr_1}{\expr_2}$, the left-to-right evaluation order is enforced by the functions $\getval$ and $\getfun$. They first evaluate their argument to a \gitree value, and then pass it on (c.f. \Cref{fig:gt:functions}).

The context-dependent reifiers for the
effects $\mathtt{callcc}$ and $\mathtt{throw}$ are defined as follows:
\begin{align*}
r_{\mathtt{callcc}} (f, (), \kappa) &= \Some(\kappa~(f~\kappa), ()) &
    r_{\mathtt{throw}} ((\alpha, f), (), \kappa) &= \Some(f~\alpha, ())
  \end{align*}

To show that the denotational semantics is sound, we need the
following lemma that shows that interpretations of expressions in
evaluation contexts are decomposed into applications of
homomorphisms.
\begin{lemma}
  \label{lem:cc:ectx_hom}
  For any context $K$ and an environment $\rho$, we have
  $\denotK{K}_\rho \in \Hom$.
For any context $K$, expression $\expr$, and an environment $\rho$,
  $\denotE{\plug{K}{\expr}}_\rho = \denotK{K}_\rho(\denotE{\expr}_\rho)$.
\end{lemma}
With these results at hand, we can show soundness of our interpretation:
\begin{lemma}{Soundness.}
  \label{lem:cc:soundness}
  Suppose $\reds{\expr_1}{\expr_2}$.
  Then $(\denotE{\expr_1}_\rho, ()) \istep^* (\denotE{\expr_2}_\rho, ())$,
  where $() \col \Tunit$ is the unique element of the unit type, representing the (lack of) state.
\end{lemma}

\paragraph{Program logic for $\iocclang$.}
We now specialize the general program logic rule
\ruleref{wp-reify-ctx-dep} using the
reifiers for $\mathtt{callcc}$ and $\mathtt{throw}$ to obtain
the following program logic rules:
\begin{mathpar}
  \inferH{wp-throw}
  {
    \begin{array}[c]{c}
      \kappa \in \Hom \quad \hasstate{(\sigma)} \arcr \later (\hasstate{(\sigma)} \wand \wpre{\eapp{f}{x}}{\Phi})
    \end{array}
  }
  {\wpre{\eapp{\kappa}{(\throwGT(\Next(x), \Next(f)))}}{\Phi}}
  \and
  \inferH{wp-callcc}
  {
    \begin{array}[c]{c}
      \kappa \in \Hom \quad \hasstate{(\sigma)} \arcr \later (\hasstate{(\sigma)} \wand \wpre{\eapp{\kappa}{(\eapp{f}{\kappa})}}{\Phi})
    \end{array}
  }
  {\wpre{\eapp{\kappa}{(\callccGT(\Next \circ f))}}{\Phi}}
\end{mathpar}
where $\kappa$ is a homomorphism representing the current evaluation context on the level of \gitrees.
The reader may wonder why these rules include the $\hasstate{(\sigma)}$ predicates, since it is just 'threaded around'.
The reason is that these rules also apply when there are other effects around and the state is composed of
different substates for different effects, cf. the discussion of modularity in \Cref{global:sec:gitrees}.

\paragraph{Adequacy and logical relation.}
Having established soundness, we now turn our attention to \emph{adequacy}, which is usually much more complicated to prove.
\begin{lemma}{Adequacy.}
  \label{lem:cc:adequacy}
  Suppose that $\typeA{\emptyset}{\expr}{\tnat}$ and
  $(\denotE{\expr}_{\emptyset}, \sigma_1) \istep^* (\Rret(n), \sigma_2)$,
  for a natural number $n$.
  Then $\contrStar{\expr}{n}$.
\end{lemma}
To prove \Cref{lem:cc:adequacy}, we define a logical relation between
syntax ($\iocclang$ programs) and semantics (\gitrees denotations)
using the program logic from \Cref{fig:cc:logrel}. To handle control
effects, we use a \emph{biorthogonal} logical
relation~\cite{Pierce2004Advanced}, adapted
from~\cite{DBLP:journals/pacmpl/TimanyB19} for adequacy, following the
Iris approach~\cite{logical-approach-to-type-soundness}.

The core observational refinement $\logrelO(\alpha, \expr)$ ensures
that if $\alpha$ reduces to a \gitree value $\ITv$, then this value is a natural number,
and $\expr$ also reduces to the same number. The evaluation context
relation $\logrelC(P)(\kappa, K)$ relates homomorphisms and evaluation
contexts when they map related arguments to expressions satisfying
$\logrelO$. The expression relation $\logrelE(P)(\alpha, \expr)$
connects related $\IT$'s and expressions in related evaluation contexts.
Types are inductively interpreted: functions relate if they map
related arguments to related results, and continuations relate via the
context relation. For open terms, the validity judgment $\Gamma \vDash
\expr \col \tau$ uses closing substitutions, with $\expr[\gamma]$
denoting applying a substitution $\gamma$ to $\expr$.

\begin{figure*}[t]
    \begin{align*}
    \logrelO(\alpha, \expr) &\eqdef \hasstate(()) \wand {\wpre{\alpha}{\beta. \Exists \val. (\contrStar{\expr}{\val}) \ast \sem{\tnat}(\beta, \val) \ast \hasstate(())}}
    \tag*{\mbox{\fbox{$\logrelO \col \erel$}}}
    \\
    \logrelC({R})(\kappa, K) &\eqdef \Box~ {\All (\beta, \val). R(\beta, \val) \wand \logrelO(\eapp{\kappa}{\beta}, \plug{K}{\val})}
    \tag*{\mbox{\fbox{$\logrelC \col \vrel \to \ectxrel$}}}
\\
    \logrelE({R})(\alpha, \expr) &\eqdef {\All (\kappa, K). \logrelC({R})(\kappa, K) \wand \logrelO(\eapp{\kappa}{\alpha}, \plug{K}{\expr})}
    \tag*{\mbox{\fbox{$\logrelE \col \vrel \to \erel$}}}
    \end{align*}
    \vspace{-0.5cm}
  \begin{align*}
    \sem{\tnat}(\alpha, \val) &\eqdef \Exists n:\mathbb{N}. \alpha = \Rret(n) \land \val = n
     \tag*{\mbox{\fbox{$\sem{\tau} \col \vrel$}}} \\
    \sem{\tarr{\tau_1}{\tau_2}}(\beta, \val) &\eqdef \Exists f. \beta = \Fun(f) \land \Box~ {\All (\alpha, \val'). \sem{\tau_1}(\alpha, \val') \wand \logrelE({\sem{\tau_2}})(\APPs{\Fun(f)}{\alpha}, \eapp{\val}{\val'})}\\
    \sem{\tcont{\tau}}(\beta, \val) &\eqdef \Exists \kappa\; K. \beta = \Fun(\Next(\Lam x. \Tick (\eapp{\kappa}{x}))) \land \val = \cont{K} \land {\logrelC({\sem{\tau}})(\kappa, K)}
  \end{align*}
  \begin{flushleft}
  \mbox{$\begin{aligned}
    \sem{\Gamma}(\rho, \gamma) &\eqdef {\All (\var \col \tau) \in \Gamma. {\sem{\tau}}(\eapp{\rho}{\var}, \eapp{\gamma}{\var})}\\
    \Gamma \vDash \expr \col \tau &\eqdef \Box~ \All (\rho, \gamma). \sem{\Gamma}(\rho, \gamma) \wand \logrelE({\sem{\tau}})(\denotE{\expr}_\rho, \plug{\expr}{\gamma})
         \end{aligned}
         $}
     \end{flushleft}
    \vspace{-1.2cm}
  \begin{flushright}
    \fbox{$
  \begin{aligned}
      \ectxrel &\eqdef {\Hom \times \Ectx} \to \iProp \\
      \vrel &\eqdef {\ITv \times \Val} \to \iProp \\
      \erel &\eqdef {\IT \times \Expr} \to \iProp
  \end{aligned}
  $}
  \end{flushright}
  \caption{Logical relation for $\iocclang$.}
  \label{fig:cc:logrel}
\end{figure*}

The proof of adequacy relies on the fact that the interpretation of
evaluation contexts are homomorphisms, which allows us to use a
limited version of the bind rule:
\begin{lemma}{Limited bind rule.}
  \label{lem:cc:obs_bind}
  If $\logrelE(P)(\alpha, \expr)$ and $\logrelC(P)(\kappa, K)$,
  then $\logrelO(\eapp{\kappa}{\alpha}, \plug{K}{\expr})$.
\end{lemma}
With this in mind we show the fundamental lemma, stating that every well-typed expression is related to its own interpretation:
\begin{lemma}{Fundamental lemma.}
  \label{lem:cc:fundamental}
  Let $\typeA{\Gamma}{\expr}{\tau}$
  then $\Gamma \vDash \expr \col \tau$.
\end{lemma}
Computational adequacy now follows easily from the fundamental lemma.
\begin{proof}[of \Cref{lem:cc:adequacy}]
  By~\Cref{lem:cc:fundamental}, we have that $\typeA{\emptyset}{\expr}{\tnat}$
  implies that $\emptyset \vDash \expr \col \tnat$. Now, the statement follows
  from~\Cref{thm:wp_adequacy} and the assumption that
    $\denotE{\expr}_{\emptyset}, \sigma_1 \istep^* \Rret~n, \sigma_2$.
\end{proof}

\section{Exceptions}
\label{global:exceptions}
Having seen how to model undelimited continuations (\texttt{call/cc} and \texttt{throw}) in \Cref{global:sec:callcc-lang}, we now turn to a more structured form of control flow: exceptions.
Exceptions represent an important middle ground between the full generality of \texttt{call/cc} and the more restricted delimited continuations we will study in \Cref{global:sec:delim-lang}.
While \texttt{call/cc} allows capturing and invoking arbitrary continuations, exceptions provide a more disciplined mechanism for non-local control transfer, where handlers are explicitly registered and exceptions can only jump to their nearest enclosing handler.

This section serves two purposes.
First, it demonstrates that our framework naturally handles exceptions, which are ubiquitous in mainstream programming languages.
Second, and more importantly, it provides a \emph{gentle introduction} to using context-dependent effects in \gitrees by walking through the design decisions involved in modeling exceptions.
We show how to decompose the operational behavior of exceptions into primitive effects, how to maintain appropriate state, and how the program logic rules follow naturally from the effect signatures.

The key insight is that exceptions require managing a \emph{stack of handlers}.
Unlike \texttt{call/cc}, which captures the entire continuation as a first-class value, exceptions maintain a stack of handler contexts that must be consulted when an exception is thrown.

We proceed as follows.
In \Cref{global:sec:exc-lang} we present $\exclang$, a lambda calculus with exceptions. 
In \Cref{global:sec:exc-sem} we show how to model exceptions in \gitrees, explaining the design choices step by step.
We conclude by deriving program logic rules for reasoning about exceptions.

\subsection{Lambda calculus with first-order exceptions, $\exclang$}
\label{global:sec:exc-lang}

We begin by defining the language we wish to model.
The language $\exclang$, shown in~\Cref{fig:exc:lang}, extends the simply-typed lambda calculus with recursive functions with two constructs for exception handling:
\begin{itemize}
  \item $\throwE{\exc}{\expr}$ raises an exception named $\exc$ with payload $\expr$;
  \item $\catch{\exc}{\expr_1}{h}{\expr_2}$ evaluates $\expr_1$ in a context where exceptions named $\exc$ are caught by the handler $h \synds{\expr_2}$.
\end{itemize}

We fix a countable set $\Exc$ of exception names with decidable equality.
Exception names allow scoping: a $\mathtt{catch}$ block only handles exceptions with a matching name, and different exception names can be handled by different handlers.
This is similar to typed exceptions in languages like Java or OCaml, where different exception types can be caught separately.

\paragraph{Operational Semantics.}
The operational semantics, shown in~\Cref{fig:exc:opsem}, uses a CEK-style abstract machine with four configurations:
\begin{itemize}
  \item $\configs{\expr}{term}$: initial configuration, about to execute the term $\expr$;
  \item $\configd{\expr}{K}{eval}$: evaluation mode, decomposing $\expr$ and building continuation stack $K$;
  \item $\configd{K}{\val}{cont}$: continuation mode, processing the stack $K$ with value $\val$;
  \item $\configs{\val}{ret}$: final configuration, returning value $\val$.
\end{itemize}

The key rules for exceptions are:
\begin{enumerate}
\item When evaluating $\catch{\exc}{\expr_1}{h}{\expr_2}$, we push a handler evaluation context $\catch{\exc}{\emptyK}{h}{\expr_2}$ onto the stack and proceed to evaluate $\expr_1$
\item If $\expr_1$ evaluates to a value without raising an exception, the handler evaluation context is simply discarded:
\[
  {\configd{\plug{K}{\catch{\exc}{\emptyK}{h}{\expr}}}{\val}{cont}} \mapsto {\configd{\val}{K}{eval}}
\]
\item If an exception $\exc$ is thrown, we search the continuation stack for the \emph{nearest} matching handler, discarding all intermediate evaluation contexts:
\[
  {\configd{\plug{K'}{\catch{\exc}{\plug{K}{\throwE{\exc}{\emptyK}}}{h}{\expr}}}{\val}{cont}} \mapsto {\configd{\subst{\expr}{h}{\val}}{K'}{eval}}
\]
where $K$ contains no handler for $\exc$.
\end{enumerate}

This operational semantics captures the essential behavior of exceptions: handlers are installed on a stack, and throwing an exception unwinds the stack to the nearest matching handler.

\begin{figure}[t]

  \begin{grammar}
\text{expressions} & \Expr \ni \expr & \val \mid \var \in \Var \mid \eapp{\expr_1}{\expr_2} \\
    \GrmContinue & \mid \catch{\exc}{\expr_1}{h}{\expr_2} \mid \throwE{\exc}{\expr} \\
    \text{values} & \Val \ni \val & \Rec f \var = \expr \\
    \text{eval. cont.} & \Ectx \ni K & \emptyK \mid \eapp{K}{\val} \mid \eapp{\expr}{K} \\
    \GrmContinue & \mid \catch{\exc}{K}{h}{\expr} \mid \throwE{\exc}{K}
  \end{grammar}
\caption{Syntax of $\exclang$.}
  \label{fig:exc:lang}
\end{figure}

\begin{figure}[t]
  \begin{align*}
    \text{abstract machine config.} \qquad \Cfg ::= \configd{\expr}{K}{eval} \mid
    \configd{K}{v}{cont} \mid
    \configs{\expr}{term} \mid
    \configs{\val}{ret}
  \end{align*}
  \begin{align*}
    {\configs{\expr}{term}} &\mapsto {\configd{\expr}{\emptyK}{eval}} \\
    {\configd{\emptyK}{\val}{cont}} &\mapsto {\configs{\val}{ret}} \\
    {\configd{\eapp{\expr_0}{\expr_1}}{K}{eval}} &\mapsto {\configd{\expr_1}{\plug{K}{\eapp{\expr_0}{\emptyK}}}{eval}} \\
    {\configd{\val}{K}{eval}} &\mapsto {\configd{K}{\val}{cont}} \\
    {\configd{\plug{K}{\eapp{\expr}{\emptyK}}}{\val}{cont}} &\mapsto {\configd{\expr}{\plug{K}{\eapp{\emptyK}{\val}}}{eval}} \\
    {\configd{\plug{K}{\eapp{\emptyK}{\expr}}}{\val}{cont}} &\mapsto {\configd{\expr}{\plug{K}{\eapp{\val}{\emptyK}}}{eval}} \\
    {\configd{\plug{K}{\eapp{(\Rec f \var = \expr)}{\emptyK}}}{\val}{cont}} &\mapsto {\configd{\subst{\subst{e}{\var}{\val}}{f}{\Rec f \var = \expr}}{K}{eval}} \\
    {\configd{\catch{\exc}{\expr_1}{h}{\expr_2}}{K}{eval}} &\mapsto {\configd{\expr_1}{\plug{K}{\catch{\exc}{\emptyK}{h}{\expr_2}}}{eval}} \\
    {\configd{\throwE{\exc}{\expr}}{K}{eval}} &\mapsto {\configd{e}{\plug{K}{\throwE{\exc}{\emptyK}}}{eval}} \\
    {\configd{\plug{K}{\catch{\exc}{\emptyK}{h}{\expr}}}{\val}{cont}} &\mapsto {\configd{\val}{K}{eval}} \\
    {\configd{\plug{K'}{\catch{\exc}{\plug{K}{\throwE{\exc}{\emptyK}}}{h}{\expr}}}{\val}{cont}} &\mapsto {\configd{\subst{\expr}{h}{\val}}{K'}{eval}} \\
    & \qquad\qquad (\langkw{try}\dots\langkw{catch}~\exc~\dots\langkw{with} \notin K)
  \end{align*}
  \caption{Operational semantics of $\exclang$ (selected rules).}
  \label{fig:exc:opsem}
\end{figure}
\subsection{Modeling Exceptions in Guarded Interaction Trees}
\label{global:sec:exc-sem}

We now show how to model $\exclang$ in \gitrees, providing step-by-step guidance on the design decisions involved.
This demonstrates the general methodology for modeling context-dependent effects.

\paragraph{Step 1: Identify the Key Operations.}
Looking at the operational semantics, we identify three key operations:
\begin{enumerate}
  \item \emph{Registering a handler}: When entering $\catch{\exc}{\expr_1}{h}{\expr_2}$, we must remember that $\exc$ can be caught by handler $h$;
  \item \emph{Normal exit}: When $\expr_1$ completes without throwing, we must remove the handler;
  \item \emph{Throwing an exception}: When $\throwE{\exc}{\expr}$ is executed, we must find the nearest handler for $\exc$ and transfer control.
\end{enumerate}

Each of these will become an effect in our model.

\paragraph{Step 2: Design the State.}
What state is needed to implement these operations?
We need to track the currently active exception handlers.
Since handlers can nest (try blocks can contain other try blocks), we use a \emph{stack} of handlers.

Each handler entry is a triple $(\exc, h, \kappa)$ where:
\begin{itemize}
  \item $\exc \in \Exc$ is the exception name this handler catches;
  \item $h : \latert \IT \to \latert \IT$ is the handler function (taking the exception payload);
  \item $\kappa : \latert \IT \to \latert \IT$ is the \emph{saved continuation} representing the computation outside the try block.
\end{itemize}

The stack is ordered with the most recent (innermost) handler at the head.
When we register a handler, we need to remember \emph{where to return} after the handler completes.
The continuation $\kappa$ represents all the computation that was ``paused'' when we entered the try block.
Without it, after handling an exception, we would not know how to resume the surrounding computation.

This is analogous to how the operational semantics uses continuation stacks: \\ the evaluation context $\catch{\exc}{K}{h}{\expr}$ remembers both the handler $h$ and the outer continuation $K$.

\paragraph{Step 3: Define the Effect Signatures.}
Based on our operations, we define three effects, shown in~\Cref{fig:exc:effects}:

\begin{itemize}
\item $\mathtt{register}$: Takes an exception name $\exc$, handler $h$, and the computation $x$ to run inside the try block.
  Returns the result of $x$.
  The reifier pushes $(\exc, h, \kappa)$ onto the handler stack, where $\kappa$ is the current continuation.
\item $\mathtt{pop}$: Takes an exception name $\exc$.
  Returns unit.
  The reifier removes the topmost handler for $\exc$ from the stack, restoring the saved continuation.
\item $\mathtt{throw}$: Takes an exception name $\exc$ and payload $x$.
  Never returns (output type is $0$, the empty type).
  The reifier searches the stack for the nearest handler for $\exc$, applies the handler to the payload, and continues with the saved continuation.
\end{itemize}

\paragraph{Step 4: Implement the Reifiers.}
The reifier functions, shown in~\Cref{fig:exc:effects}, directly implement the stack operations:

\begin{itemize}
  \item $r_{\mathtt{register}}$ pushes a new handler onto the stack and returns the inner computation;
  \item $r_{\mathtt{pop}}$ removes the handler from the stack if it matches the given exception name (failing otherwise), and continues with both the current and saved continuations; 
  \item $r_{\mathtt{throw}}$ searches for the first matching handler, applies it to the payload, and continues with the saved continuation, \emph{discarding} the current continuation.
\end{itemize}

The fact that $r_{\mathtt{throw}}$ searches the stack and discards intermediate continuations is what gives exceptions their non-local control transfer behavior.

\paragraph{Step 5: Define Convenient Wrapper Functions.}
We define high-level wrapper functions that combine the primitive effects:

\begin{itemize}
\item $\THROW(\exc, x)$: First evaluates $x$ to a value, then performs $\mathtt{throw}$;
\item $\CATCH(\exc, h, x)$: Performs $\mathtt{register}$, evaluates $x$, then performs $\mathtt{pop}$ to clean up the handler.
\end{itemize}

The $\CATCH$ function shows a key pattern: we register the handler, run the protected computation, and then pop the handler on normal exit.
If an exception is thrown during the computation, the $\mathtt{throw}$ effect will handle it, and the $\mathtt{pop}$ will never be reached.

\paragraph{Step 6: Define the Denotational Semantics.}
The interpretation of exception constructs, shown in~\Cref{fig:exc:denot}, is now straightforward.

Note how the handler $h \synds{\expr_2}$ becomes a lambda function that interprets $\expr_2$ with $h$ bound to the exception payload.

\begin{figure*}[t]
  \begin{flushleft}
    \mbox{\fbox{$\stateO \eqdef \mathsf{list}~(\Exc \times (\latert \IT \to \latert \IT) \times (\latert \IT \to \latert \IT))$}}
  \end{flushleft}
  \begin{align*}
    \Ins_{\mathtt{throw}}(X) & \eqdef \Exc \times \latert X & \Outs_{\mathtt{throw}}(X) & \eqdef 0 \\
    \Ins_{\mathtt{register}}(X) & \eqdef \Exc \times (\latert X \to \latert X) \times \latert X & \Outs_{\mathtt{register}}(X) & \eqdef \latert X \\
    \Ins_{\mathtt{pop}}(X) & \eqdef \Exc & \Outs_{\mathtt{pop}}(X) & \eqdef 1 \\
  \end{align*}
  \begin{align*}
    r_{\mathtt{register}}((\exc, h, x), \sigma, \kappa) &= \Some(x, (\exc, h, \kappa) \cc \sigma) \\
    r_{\mathtt{throw}}((\exc, x), \sigma, \kappa) &=
                                                 \begin{cases}
                                                   \Some(\kappa'~(h~x), \sigma_2), & \sigma = \sigma_1 \app [(\exc, h, \kappa')] \app \sigma_2 \\
                                                   \None, & \text{otherwise}
                                                 \end{cases}  \\
    r_{\mathtt{pop}}(\exc, \sigma, \kappa) &=
                                          \begin{cases}
                                            \None, & \sigma = [] \\
                                            \None, & \sigma = (\exc', h, \kappa') \cc \sigma' \text{ and } \exc \neq \exc' \\
                                            \Some(\kappa'~(\kappa~()), \sigma'), & \sigma = (\exc, h, \kappa') \cc \sigma'
                                          \end{cases}
  \end{align*}
  \begin{align*}
    \POP(\exc) & \eqdef \Vis_{\mathtt{pop}}(\exc, \lambda \_. \Next (\Rret())) \\
    \REG(\exc, h, x) & \eqdef \Vis_{\mathtt{register}}((\exc, h, x), \idfun) \\
    \THROW(\exc, x) & \eqdef \getval(x, \lambda r. \Vis_{\mathtt{throw}}((\exc, \Next(r)), \bot\text{-}\mathsf{elim})) \\
    \CATCH(\exc, h, x) & \eqdef \REG(\exc, h, \Next(\getval(x, \lambda r. \getval(\POP(\exc), \lambda \_. r))))
  \end{align*}
  \caption{Effects for $\exclang$.}
  \label{fig:exc:effects}
\end{figure*}

\begin{figure}[t]
  \begin{align*}
    &\denotE{\throwE{\exc}{\expr}}_{\rho} = \THROW(\exc, \denotE{\expr}_{\rho})
    \\
    &\denotE{\catch{\exc}{\expr_1}{h}{\expr_2}}_{\rho} = \CATCH(\exc, \lambda v. \denotE{\expr_2}_{\rho[h \mapsto v]}, \denotE{\expr_1}_{\rho})
  \end{align*}
  \caption{Denotational semantics for $\exclang$ (exception constructs).}
  \label{fig:exc:denot}
\end{figure}

We then proceed to lift $\denotE{-}$ to evaluation contexts, $\denotK{-}$, and the abstract machine's configurations, $\denotS{-}$, in~\Cref{fig:exc:denotAbs}.
\begin{figure}
  \begin{align*}
      &\denotK{\emptyK}_\rho = (\Lam \var. \var, []) \\ 
      &\denotK{\eapp{K}{\val}}_\rho = \Let (\kappa, \sigma) = \denotK{K}_\rho in \left(\kappa \circ (\Lam \var. \APPs{\var}{(\denotE{\val}_\rho)}), \sigma\right) \\ 
      &\denotK{\eapp{\expr}{K}}_\rho = \Let (\kappa, \sigma) = \denotK{K}_\rho in \left(\kappa \circ (\Lam \var. \APPs{(\denotE{\expr}_\rho)}{\var}), \sigma\right) \\ 
      &\denotK{\throwE{\exc}{K}}_\rho = \Let (\kappa, \sigma) = \denotK{K}_\rho in \left(\kappa \circ (\Lam \var. \THROW(\exc, \var)), \sigma\right) \\ 
      &\denotK{\catch{\exc}{K}{h}{\expr}}_\rho = \Let (\kappa, \sigma) = \denotK{K}_\rho in \\ 
      & \qquad\qquad \left(\Lam \var. \getval(\var, \Lam y. \getval(\POP(\exc), \Lam _. y)), (\exc, \latert(\Lam \var. \denotE{\expr}_{h \mapsto \var}), \latert \kappa) \cc \sigma\right)
  \end{align*}

  \begin{align*}
      &\denotS{\configd{\expr}{K}{eval}}_{\rho} = \Let (\kappa, \sigma) = \denotK{K}_\rho in \left(\kappa(\denotE{\expr}_{\rho}), \sigma\right) 
      \\
      &\denotS{\configd{K}{\val}{cont}}_{\rho} = \Let (\kappa, \sigma) = \denotK{K}_\rho in \left(\kappa(\denotE{\val}_{\rho}), \sigma\right) \\
      &\denotS{\configs{\expr}{term}}_{\rho} = \left(\denotE{\expr}_{\rho}, [] \right) \\
      &\denotS{\configs{\val}{ret}}_{\rho} = \left(\denotE{\val}_{\rho}, [] \right)
  \end{align*}
  \caption{Denotational semantics of the abstract machine's configurations (selected clauses).}
  \label{fig:exc:denotAbs}
\end{figure}

\paragraph{Step 7: Prove Soundness.}
We prove that the denotational semantics is sound with respect to the operational semantics:

\begin{theorem}[Soundness]\label{lem:exc:soundness}
  Let $c_0, c_1 \in \Cfg$ and suppose $\reds{c_0}{c_1}$.
  Then $\denotS{c_{0}}_{\rho} \istep^* \denotS{c_{1}}_{\rho}$ for all $\rho$.
\end{theorem}

The proof proceeds by induction on the operational semantics derivation, using the program logic rules derived in the next paragraph.

\paragraph{Step 8: Derive Program Logic Rules.}
From the effect signatures and reifiers, we derive program logic rules for reasoning about exceptions, shown in~\Cref{fig:exc:wpre}.
These rules follow directly from the \ruleref{wp-reify} rule for general effects.
The key rules are:
\begin{itemize}
\item \ruleref{wp-reg}: Registering a handler requires the state $\sigma$ and produces state with the handler pushed: $(\exc, h, \latert \kappa) \cc \sigma$;
\item \ruleref{wp-throw'}: Throwing requires that a matching handler exists in the state, and transfers control to that handler;
\item \ruleref{wp-catch-throw}: A combined rule for the common case where an exception is immediately thrown and caught within the same $\CATCH$.
\end{itemize}

These rules enable modular reasoning about exception-using code within the program logic.

\begin{figure}[t]
  \begin{mathpar}
    \inferH{wp-reg}
    {
      \begin{array}[c]{c}
        \hasstate{(\sigma)} \arcr
        \later (\hasstate{((\exc, h, \latert \kappa) \cc \sigma)} \wand \wpre{\beta}{\Phi}) \arcr
\end{array}
    }
    {\wpre{\kappa(\REG(\exc, h~(\Next(\beta))))}{\Phi}}
    \and
    \inferH{wp-pop'}
    {
      \begin{array}[c]{c}
        \hasstate{((\exc, h, \kappa') \cc \sigma)} \arcr
        \later (\hasstate{(\sigma)} \wand \wpre{\beta}{\Phi}) \arcr
        \kappa' (\latert \kappa~(\Next(\Rret()))) = \Next (\beta)
      \end{array}
    }
    {\wpre{\kappa(\POP(\exc))}{\Phi}}
    \and
    \inferH{wp-catch}
    {
      \begin{array}[c]{c}
        \hasstate{(\sigma)} \arcr
        x \in \Val \arcr
        \later (\hasstate{(\sigma)} \wand \wpre{\kappa~x}{\Phi})
      \end{array}
    }
    {\wpre{\kappa(\CATCH(\exc, h, x))}{\Phi}}
    \and
    \inferH{wp-throw'}
    {
      \begin{array}[c]{c}
        \exc \notin \sigma' \arcr 
        \hasstate{(\sigma' \app [(\exc, h, \kappa')] \app \sigma)} \arcr
        \later (\hasstate{(\sigma)} \wand \wpre{\beta}{\Phi}) \arcr
        \kappa' (h~(\Next(x))) = \Next (\beta)
      \end{array}
    }
    {\wpre{\kappa(\THROW(\exc, x))}{\Phi}}
    \and
    \inferH{wp-catch-throw}
    {
      \begin{array}[c]{c}
        \hasstate{(\sigma)} \arcr
        \later (\hasstate{(\sigma)} \wand \wpre{\kappa~\beta}{\Phi}) \arcr
        h~(\Next(x)) = \Next (\beta)
      \end{array}
    }
    {\wpre{\kappa(\CATCH(\exc, h, (\kappa'(\THROW(\exc, x)))))}{\Phi}}
  \end{mathpar}
  \caption{Weakest precondition rules for exceptions.}
  \label{fig:exc:wpre}
\end{figure}

\paragraph{Summary.}
This completes our treatment of exceptions.
We have shown how to systematically design effects for exceptions by:
\begin{enumerate}
\item Identifying the key operations from the operational semantics;
\item Designing appropriate state (a handler stack with saved continuations);
\item Defining effect signatures that capture these operations;
\item Implementing reifiers that manipulate the state correctly;
\item Deriving high-level wrappers and program logic rules.
\end{enumerate}

This methodology generalizes to other context-dependent effects.
The key is understanding what state needs to be maintained and how the effects should manipulate it.
In the next section, we apply similar techniques to the more complex case of delimited continuations.

\section{Modeling Delimited Continuations}
\label{global:sec:delim-lang}
In this section we show how guarded interaction trees can be used to give semantics to delimited continuations which are a challenging example of context-dependent effects.
We provide a denotational semantics for a programming language $\delimlang$ with \texttt{shift} and \texttt{reset} effects.
Our approach here for giving denotational semantics to $\delimlang$ is similar to what we did for $\exclang$ in~\Cref{global:exceptions}.
However, in addition to soundness our denotational semantics (\Cref{lem:exc:soundness}), we also show computational adequacy for $\delimlang$ relative to an abstract machine semantics \citep{DBLP:journals/lmcs/BiernackaBD05}.
To the best of our knowledge, this represents the first formalized sound and adequate direct-style denotational semantics for delimited continuations.
\
\subsection{Syntax and Operational Semantics of $\delimlang$}
\begin{figure}[t]
  \begin{flushright}
    \fbox{$
      \begin{aligned}
        &\typeD{\Gamma}{\expr}{\tau}{\alpha}{\beta}\\
        &\typeDpure{\Gamma}{\expr}{\tau}
      \end{aligned}$}
  \end{flushright}
  \vspace{-1.0cm}
  \begin{grammar}
    \text{types} & \Tys \ni \tau, \sigma, \alpha, \beta, \delta, \gamma & \tnat \mid \tarr{\tau/\alpha}{\sigma/\beta} \mid \tcont{\tau, \alpha} \\
    \text{expressions} & \Expr \ni \expr & \val \mid \var \mid \eapp{\expr_1}{\expr_2} \mid \natop{\expr_1}{\expr_2} \\
    \GrmContinue & \mid \If \expr_1 then \expr_2 \Else \expr_3 \mid \control{\var}{\expr} \mid \delim{\expr} \mid \ekapp{\expr_1}{\expr_2} \\
    \text{values} & \Val \ni \val & n \mid \Rec f \var = \expr \mid \cont{K} \\
    \text{eval. contexts} & \Ectx \ni K & \emptyK \mid \plug{K}{\If \emptyK then \expr_1 \Else \expr_2}\mid \plug{K}{\eapp{\val}{\emptyK}} \mid \plug{K}{\eapp{\emptyK}{\expr}} \\
    \GrmContinue & \mid \plug{K}{\natop{\expr}{\emptyK}} \mid \plug{K}{\natop{\emptyK}{\val}} \mid \plug{K}{\ekapp{\emptyK}{\val}} \mid \plug{K}{\ekapp{\expr}{\emptyK}}
  \end{grammar}
  \vspace{-0.2cm}
  \begin{mathpar}
    \inferrule{\typeDpure{\Gamma}{\expr}{\tau}}{\typeD{\Gamma}{\expr}{\tau}{\alpha}{\alpha}}
    \and
    \inferrule{\typeD{\Gamma, \var \col \cont{(\tau, \alpha)}}{\expr}{\sigma}{\sigma}{\beta}}{\typeD{\Gamma}{\control{\var}{\expr}}{\tau}{\alpha}{\beta}}
    \and
    \inferrule{\typeD{\Gamma}{\expr}{\tau}{\tau}{\sigma}}{\typeDpure{\Gamma}{\delim{\expr}}{\sigma}}
    \and
    \inferrule{\var \col \tau \in \Gamma}{\typeDpure{\Gamma}{\var}{\tau}}
    \and
    \inferrule{\typeD{\Gamma, f \col \tarr{\sigma/\alpha}{\tau/\beta}, \var \col \sigma}{\expr}{\tau}{\alpha}{\beta}}{\typeDpure{\Gamma}{\Rec f \var = \expr}{\tarr{\sigma/\alpha}{\tau/\beta}}}

    \and
    \inferrule{\begin{array}[c]{c}\typeD{\Gamma}{\expr_1}{\tarr{\sigma/\alpha}{\tau/\beta}}{\gamma}{\delta} \arcr \typeD{\Gamma}{\expr_2}{\sigma}{\beta}{\gamma}\end{array}}{\typeD{\Gamma}{\eapp{\expr_1}{\expr_2}}{\tau}{\alpha}{\delta}}
    \and
    \inferrule{\typeD{\Gamma}{\expr_1}{\tnat}{\beta}{\alpha} \quad \typeD{\Gamma}{\expr_2}{\tau}{\sigma}{\beta} \quad \typeD{\Gamma}{\expr_3}{\tau}{\sigma}{\beta}}{\typeD{\Gamma}{\If \expr_1 then \expr_2 \Else \expr_3}{\tau}{\sigma}{\alpha}}
    \and
    \inferrule{\\}{\typeDpure{\Gamma}{n}{\tnat}}
    \and
    \inferrule{\typeD{\Gamma}{\expr_1}{\tnat}{\alpha}{\beta} \quad \typeD{\Gamma}{\expr_2}{\tnat}{\beta}{\sigma}}{\typeD{\Gamma}{\natop{\expr_1}{\expr_2}}{\tnat}{\alpha}{\sigma}}
    \and
    \inferrule{\typeD{\Gamma}{\expr_1}{\cont{(\tau, \alpha)}}{\sigma}{\delta} \quad \typeD{\Gamma}{\expr_2}{\tau}{\delta}{\beta}}{\typeD{\Gamma}{\ekapp{\expr_1}{\expr_2}}{\alpha}{\sigma}{\beta}}
  \end{mathpar}
  \caption{Syntax and typing rules of $\delimlang$.}
  \label{fig:d:syn_delim_control}
\end{figure}
\begin{figure}[t]

\begin{minipage}[t]{0.45\textwidth}
    \begin{align*}
      {\configs{\expr}{term}} &\mapsto {\config{\expr}{\emptyK}{[]}{eval}} \\
      {\configd{K \cc mk}{\val}{mcont}} &\mapsto {\config{K}{\val}{mk}{cont}} \\
      {\configd{[]}{\val}{mcont}} &\mapsto {\configs{\val}{ret}} \\
      {\config{\emptyK}{\val}{mk}{cont}} &\mapsto {\configd{mk}{\val}{mcont}} \\
      {\config{\plug{K}{\ekapp{\emptyK}{\val}}}{\cont{K'}}{mk}{cont}} &\mapsto {\config{K'}{\val}{K \cc mk}{cont}} \\
      {\config{\plug{K}{\ekapp{\expr}{\emptyK}}}{\val}{mk}{cont}} &\mapsto {\config{\expr}{\plug{K}{\ekapp{\emptyK}{\val}}}{mk}{eval}} \\
{\config{\val}{K}{mk}{eval}} &\mapsto {\config{K}{\val}{mk}{cont}} \\
{\config{\ekapp{\expr_0}{\expr_1}}{K}{mk}{eval}} &\mapsto {\config{\expr_1}{\plug{K}{\ekapp{\expr_0}{\emptyK}}}{mk}{eval}} \\
      {\config{\delim{\expr}}{K}{mk}{eval}} &\mapsto {\config{\expr}{\emptyK}{K \cc mk}{eval}} \\
      {\config{\control{k}{\expr}}{K}{mk}{eval}} &\mapsto {\config{\subst{e}{k}{K}}{\emptyK}{mk}{eval}}
    \end{align*}
  \end{minipage}
  \begin{minipage}[t]{0.3\textwidth}
    \begin{align*}
      \text{metacontinuations:}\\
      \Mcont \ni mk ::= [] \mid K \cc mk \\
      \text{abstract machine config.:}\\
      \Cfg ::= \config{\expr}{K}{mk}{eval} \mid\\
      \config{K}{v}{mk}{cont} \mid \\
      \configd{mk}{v}{mcont} \mid \\
      \configs{\expr}{term} \mid\\
      \configs{\val}{ret}
    \end{align*}
\end{minipage}
  \caption{Operational semantics of $\delimlang$ (excerpt).}
  \label{fig:d:dyn_sem_delim_control}
\end{figure}

The syntax and the type system of $\delimlang$ is given in \Cref{fig:d:syn_delim_control}.
It is similar to $\iocclang$, but instead of $\callcc{-}{-}$ there are operators $\delim{\expr}$ (delimit the current evaluation context, also known as \texttt{reset}) and $\control{\var}{\expr}$ (grab the current delimited continuation, and bind it to $\var$ in $\expr$, also known as \texttt{shift}).

The type system follows Danvy and Filinski \cite{danvy.filinski.1989}, extending simply-typed $\lambda$-calculus with answer types $\alpha$, $\beta$.
The main typing judgment $\typeD{\Gamma}{\expr}{\tau}{\alpha}{\beta}$ means: under the typing context $\Gamma$, expression $\expr$ can be plugged into a context expecting a value of type $\tau$ and producing a value of type $\alpha$; in that case the resulting program will have the type $\beta$.
You can think of it a computation $\Gamma \to (\tau \to \alpha) \to \beta$ under the CPS translation.
Thus, the type of the (delimited) continuation corresponds to $\tau \to \alpha$, while the type of the overall expression is $\beta$.
The pure typing judgment $\typeDpure{\Gamma}{\expr}{\tau}$ indicates $\expr$ does not depend on the surrounded context, and is context-independent for any answer types.
Expressions can change their context's answer type, as seen in the $\control{\var}{\expr}$ typing rule.

For example, suppose we extend the type system with booleans, $\Tbool$, and add a primitive function $\mathsf{isprime}$ that does not modify answer types.
That is $\typeD{\emptyset}{\mathsf{isprime}}{\tarr{\tnat/\beta}{\Tbool/\beta}}{\alpha}{\alpha}$.
The expression ${\delim{(\eapp{(\Rec f x = \eapp{\mathsf{isprime}}{(\control{k}{x - 1})})}{2})}}$ is well-typed in this type system as a $\tnat$, even though it changes the answer type from $\Tbool$ to $\tnat$.

Answer types appear in both judgments and type constructors.
Continuation type $\cont(\tau,\alpha)$ represents contexts expecting something of the type $\tau$ and producing something of the type $\alpha$.
Function type $\tarr{\sigma/\alpha}{\tau/\beta}$, in addition to the input type $\sigma$ and the output type $\tau$, record the typing of the surrounding context at the point of the function call.
See \cite{danvy.filinski.1989} for details.

The operational semantics for $\delimlang$, similarly to $\exclang$, uses a CEK machine, following \cite{DBLP:journals/lmcs/BiernackaBD05,DBLP:conf/ifip2/FelleisenF87,DBLP:conf/ppdp/BiernackaB09}.
Selected reduction rules appear in \Cref{fig:d:dyn_sem_delim_control} (see Rocq formalization for the full set of rules).
The abstract machine operates on various \emph{configurations}, which can be of several forms.
The first one is the initial configuration $\configs{\expr}{term}$, which is just a starting state for evaluating expressions.
Similarly, there is a terminal configuration $\configs{\val}{ret}$ signifying that the program has terminated with the value $\val$.

From the initial configuration, we go on to $\config{\expr}{K}{mk}{eval}$, which signifies that we are evaluating an expression $\expr$ inside the current delimited context $K$, with the metacontinuation $mk$ (a stack of continuations based on different delimiters).
It is this configuration type which takes care of delimited control operations.
The $\delim{}$ operator saves the current continuation on top of the metacontinuation, limiting the scope of $\control{\var}{\expr}$.
The $\control{\var}{\expr}$ operation behaves similarly to $\callcc{\var}{\expr}$, except that it prevents later control operators from capturing its evaluation context.

The last two configuration types are for dealing with continuations and metacontinuations.
A configuration $\config{K}{\val}{mk}{cont}$ signifies that we are trying to plug in the value $\val$ into the context $K$, with the metacontinuation $mk$.
A configuration $\configd{mk}{v}{mcont}$ signifies that we are done with the current continuation (ending with the value $v$), but we still have to unwind the continuation stack $mk$.

{
  \subsection{Denotational Semantics of $\delimlang$}
  \label{global:sec:delim-sem}
  \begin{figure*}[t]
\begin{align*}
  \Ins_{\mathtt{reset}}(X) & \eqdef \latert X & \Ins_{\mathtt{shift}}(X) & \eqdef (\latert X \rightarrow \latert X) \rightarrow \latert X \\
  \Ins_{\mathtt{pop}}(X) & \eqdef \latert X & \Ins_{\mathtt{appcont}}(X) & \eqdef \latert X \times \latert (X \rightarrow X) \\
  \Outs_{\mathtt{reset}}(X) & \eqdef \latert X & \Outs_{\mathtt{shift}}(X) & \eqdef \latert X \\ \Outs_{\mathtt{pop}}(X) & \eqdef 0 & \Outs_{\mathtt{appcont}}(X) & \eqdef \latert X \\
      r_{\mathtt{reset}}(e, \sigma, \kappa) &= \Some(e, \kappa \cc \sigma) &
                                                                           r_{\mathtt{shift}}(f, \sigma, \kappa) &= \Some(\eapp{f}{\kappa}, \sigma) \\
  r_{\mathtt{pop}}(e, [], \_) &= \Some(e, []) &
                                                r_{\mathtt{pop}}(e, \kappa \cc \sigma, \_) &= \Some(\eapp{\kappa}{e}, \sigma) \\
   r_{\mathtt{appcont}} ((e, \kappa), \sigma, \kappa') &= \Some(\eapp{\kappa}{e}, \kappa' \cc \sigma) & \popAux(\beta) &\eqdef \getval(\beta, \POP) \\
  \RESET(e) &\eqdef \Vis_{\mathtt{reset}}(e, \idfun) &
                                                       \SHIFT(f) &\eqdef \Vis_{\mathtt{shift}}(f, \idfun) \\
  \APPCONT(e, f) &\eqdef \Vis_{\mathtt{appcont}}{((e, f), \idfun)} & \POP(e) &\eqdef \Vis_{\mathtt{pop}}{(e, \Lam x. \elimF x)}
\end{align*}
\caption{Effects for $\delimlang$.}
\label{fig:delim:effects}
\end{figure*}

Our model represents delimited continuations with effects mimicking an
abstract machine, operating on semantic rather than syntactic
components. The effect signature and reifiers
(\Cref{fig:delim:effects}) define a state with a stack of
continuations, manipulated explicitly. The effect signature
$E_{\delimlang}$ includes four operators:
$
 \{\mathtt{reset}, \mathtt{shift}, \mathtt{pop}, \mathtt{appcont}\}.
$
The signature of $\mathtt{reset}$ simply tells us that the
corresponding effect does not directly modify its argument. The
auxiliary effect $\mathtt{pop}$, which does not have an equivalent in
the surface syntax, is used to enforce unwinding of the continuation
stack. As the output arity of $\mathtt{pop}$ signifies, it does not
return. We describe the importance of that below. The rest of the
signature is more straightforward: $\mathtt{shift}$ and
$\mathtt{appcont}$ are defined exactly as $\mathtt{callcc}$ and
$\mathtt{throw}$. The semantics of these effects, in terms of
reification, is more intricate. As we mentioned, the state for
reification is $\State = \List(\latert \IT \to \latert \IT)$.

In comparison with $\callcc{\var}{\expr}$, the control operator $\control{\var}{\expr}$ does not necessarily continue from the same continuation; hence, the corresponding reifier passes the current continuation to the body, but does not return control back.
The reifier for $\mathtt{reset}$ simply saves the current continuation $\kappa$ onto the stack $\sigma$.
It is then the job of the $\mathtt{pop}$ operation to restore the continuation from the stack.
The reifier for $\mathtt{shift}$ is similar to that of $\mathtt{callcc}$, except that it removes the current continuation entirely.
The reifier for $\mathtt{appcont}$, in comparison with $\mathtt{throw}$,
does not simply pass control, but also saves the current continuation
on the stack.
This corresponds to the fact that whenever a delimited continuation is invoked, the result is wrapped in a reset; that is done to prevent the continuation from escaping the delimiter.
\begin{figure}[t]
  \begin{mathpar}
    \inferH{wp-shift}
    {
      \begin{array}[c]{c}
        \hasstate{(\sigma)} \arcr
        \later (\hasstate{(\sigma)} \wand \wpre{\beta}{\Phi}) \and\arcr
        \latert \popAux (f (\latert \kappa)) = \Next(\beta)
      \end{array}
    }
    {\wpre{\kappa(\SHIFT(f))}{\Phi}}
    \and
    \inferH{wp-reset}
    {
      \begin{array}[c]{c}
        \hasstate{(\sigma)} \arcr
        \later (\hasstate{(\latert \kappa \cc \sigma)} \wand \wpre{\popAux(e)}{\Phi})
      \end{array}
    }
    {\wpre{\kappa (\RESET (\Next (e)))}{\Phi}}
    \and
    \inferH{wp-pop}
    {
      \begin{array}[c]{c}
        \hasstate{(\sigma)} \arcr
        \kappa' = \kappa~\text{if $\sigma = \kappa \cc \sigma'$ and $\idfun$ otherwise} \arcr
        \later (\hasstate{(\mathtt{tail}(\sigma))} \wand \wpre{\kappa'(v)}{\Phi})
      \end{array}
    }
    {\wpre{\popAux(v)}{\Phi}}
\and
    \inferH{wp-appcont}
    {
      \begin{array}[c]{c}
        \hasstate{(\sigma)} \arcr
        \later (\hasstate{(\latert \kappa \cc \sigma)} \wand \wpre{\beta}{\Phi}) \arcr
        \latert \kappa' (e) = \Next (\beta)
      \end{array}
    }
    {\wpre{\kappa(\APPCONT(e, \kappa'))}{\Phi}}
  \end{mathpar}
  \caption{Weakest precondition rules for delimited continuations.}
  \label{fig:d:wpre}
\end{figure}
As part of instantiating \gitrees with these effects, we obtain the specialized program logic rules shown in \Cref{fig:d:wpre}.
We will use those rules later for defining a logical relation between the syntax and the semantics of $\delimlang$.
\begin{figure}[t]
  \begin{align*}
    &\denotE{\delim{\expr}}_{\rho} = \RESET (\popAux (\denotE{\expr}_\rho)) \\
    &\denotE{\control{\var}{\expr}}_{\rho} =  \SHIFT (\popAux \circ (\Lam \kappa. \denotE{\expr}_{\rho, \var \mapsto \Fun (\Next (\Lam y. \Tau (\kappa (\Next y))))})) \\
    &\denotE{\ekapp{\expr_1}{\expr_2}}_{\rho} = \getval (\denotE{\expr_2}_{\rho}, \Lam x. \getfun (\denotE{\expr_1}_\rho, \Lam y. \APPCONT (\Next(x), y))) \\
    &\denotV{\cont{K}}_\rho = \Fun(\Next (\Lam x. \Tick(\popAux (\denotK{K}_\rho\;x)))) \\
    &\denotK{\plug{K}{\ekapp{\emptyK}{\val}}}_\rho = \Lam \var. \denotK{K}_\rho (\denotE{\ekapp{\var}{\val}}_\rho) \\
    &\denotM{mk}_{\rho} = \mathrm{map} (\Lam k. \popAux \circ \denotK{k}_{\rho}) mk \\
    &\denotS{\config{\expr}{K}{mk}{eval}}_{\rho} = \left(\popAux (\denotE{\plug{K}{\expr}}_{\rho}), \denotM{mk}_{\rho}\right) \\
    &\denotS{\config{K}{\val}{mk}{cont}}_{\rho} = \left(\popAux (\denotE{\plug{K}{\val}}_{\rho}), \denotM{mk}_{\rho}\right) \\
    &\denotS{\configd{mk}{\val}{mcont}}_{\rho} = \left( \popAux (\denotV{\val}_{\rho}), \denotM{mk}_{\rho} \right) \\
    &\denotS{\configs{\expr}{term}}_{\rho} = \left(\popAux (\denotE{\expr}_{\rho}), [] \right) \\
    &\denotS{\configs{\val}{ret}}_{\rho} = \left( \denotV{\val}_{\rho}, [] \right)
  \end{align*}
  \caption{Denotational semantics for a calculus with delimited control (selected clauses).}
  \label{fig:d:delim_sem}
\end{figure}
As mentioned above, we will use $\POP$ to unwind the continuation stack and restore the continuation after finishing with a $\mathtt{reset}$.
This means that we will need to insert explicit calls to $\POP$ in the interpretation of $\delimlang$.
For these purposes, we use an abbreviation $\popAux(\beta)$, which first evaluates $\beta$ to a value, and then executes the $\mathtt{pop}$ operation.

The interpretation of $\delimlang$ uses this auxiliary function and is given in \Cref{fig:d:delim_sem}.
Similarly to the operational semantics, the interpretation is divided into five categories.
First, we have $\denotE{-}$ and $\denotV{-}$ for the interpretation of expressions and values, which is what we need for the surface syntax.
All of those interpretations return \gitrees{}.
Note that in the interpretation of $\delim{-}$ we insert explicit calls to $\popAux$, and similarly in the interpretation of continuations.

The other group of interpretations, $\denotK{-}$, $\denotM{-}$ and $\denotS{-}$, are for interpreting continuations, metacontinuations, and other configurations;
these are used for showing soundness (preservation of operational semantics by the interpretation).
The interpretation $\denotK{-}$ of continuations returns a semantic continuation (a function $\IT \to  \IT$).
Similarly, the interpretations $\denotM{-}$ (resp. $\denotS{-}$) of metacontinuations (resp. configurations) returns a stack of semantic continuations (resp. a semantic configuration).

We now show that our interpretation is sound w.r.t. the abstract machine semantics.
For this we prove lemmas similar to \Cref{lem:cc:ectx_hom}, and put them to use in the soundness theorem:
\begin{theorem}{Soundness.}
  \label{lem:d:soundness}
  Let $c_0, c_1 \in \Cfg$ and suppose $\reds{c_0}{c_1}$.
  Then $\denotS{c_{0}} \istep^* \denotS{c_{1}}$.\end{theorem}

\subsection{Logical Relation and Adequacy}
We now show that our denotational semantics is adequate with regards to the abstract machine semantics.
Specifically, we show the following result:
\begin{theorem}{Adequacy.}
  \label{lem:d:adequacy}
  Suppose $\typeD{\emptyset}{\expr}{\tnat}{\tnat}{\tnat}$ is a well-typed term,
  and that $\left(\popAux(\denotE{\expr}_{\emptyset}), []\right) \istep^* \left(\Rret(n), \sigma\right)$ for a natural number $n$ and a metacontinuation $\sigma$.
  Then $\contrStar{\configs{\expr}{term}}{\configs{n}{ret}}$.
\end{theorem}
\begin{figure*}[t]
  \begin{minipage}{0.45\textwidth}
    \begin{align*}
      \mathrm{SynConf} &\eqdef \Expr \times \Ectx \times \Mcont \\
      \mathrm{SemConf} &\eqdef \IT \times \Hom \times \List(\Hom) \\
      \confrel &\eqdef {\mathrm{SemConf} \times \mathrm{SynConf}} \to \iProp \\
      \mrel &\eqdef {\textdom{list}~\Hom \times \Mcont} \to \iProp \\
      \ectxrel &\eqdef {\Hom \times \Ectx} \to \iProp \\
            \erel &\eqdef {\IT \times \Expr} \to \iProp
    \end{align*}
  \end{minipage}
  \begin{minipage}{0.45\textwidth}
    \begin{align*}
      \vrel &\eqdef {\ITv \times \Val} \to \iProp \\
      \logrelO &\col \confrel \\
      \logrelM &\col \vrel \to \mrel \\
      \logrelC &\col \vrel \to \vrel \to \ectxrel \\
      \logrelE &\col \vrel \to \vrel \to \vrel \to \erel  \\
      \sem{\tau} &\col \vrel
    \end{align*}
  \end{minipage}
  \vspace{-0.1cm}
  \begin{align*}
    \logrelO\big((\alpha,\kappa,\sigma),(e,K,mk)\big) &\eqdef
               \begin{array}[c]{l}
                 \hasstate{(\sigma)} \wand{}\\\quad \textlog{wp}\spac{\popAux (\eapp{\kappa}{\alpha})}\spac\big\{\Ret \beta. \Exists \val. (\contrStar{\config{e}{K}{mk}{eval}}{\configs{\val}{ret}}) \\\quad \ast (\beta, \val) \in \sem{\tnat} \ast \hasstate{([])}\big\}
               \end{array} \\
    \logrelM({P})(\sigma, mk) &\eqdef {\All (\alpha, \val). P(\alpha, \val) \wand \logrelO((\alpha, \imath, \sigma), (\val, \emptyK, mk))} \\
    \logrelC({Q}, {P})(\kappa, K) &\eqdef \begin{array}[c]{l}
                                            \Box~\forall (\alpha, \val). Q(\alpha, \val) \wand \forall (\sigma, mk). \logrelM({P})(\sigma, mk) \wand \\\quad \logrelO((\alpha, \kappa, \sigma), (\val, K, mk))
                                            \end{array} \\
    \logrelE({P, Q, R})(\beta, \expr) &\eqdef \begin{array}[c]{l}
                                                \All (\kappa, K). \logrelC({P}, {Q})(\kappa, K) \wand \All (\sigma, mk). \logrelM({R})(\sigma, mk) \wand \\\quad \logrelO((\beta, \kappa, \sigma), (\expr, K, mk))
                                              \end{array}\\
  \end{align*}
  \vspace{-1.2cm}
  \begin{align*}
    \sem{\tnat}(\beta, \val) & \eqdef \Exists n \in \mathbb{N}. \beta = \Rret(n) \land \val = n \\
    \sem{\tarr{\tau/\alpha}{\sigma/\beta}}(\theta, \val) & \eqdef \begin{array}[c]{l}
                                                                    \Exists F. \theta = \Fun(F) \land \Box~ \All (\eta, w). \sem{\tau}(\eta, w) \wand \\ \quad \logrelE(\sem{\sigma},\sem{\alpha},\sem{\beta})(\APPs{\theta}{\eta},\eapp{\val}{w})
                                                                  \end{array}\\
\sem{\tcont{\tau, \alpha}}(\beta, \val) & \eqdef \begin{array}[c]{l}
                                                       \Exists \kappa\; K. \beta = \Fun (\Next (\Lam x. \Tick (\eapp{(\popAux \circ \kappa)}{x}))) \land \val = \cont{K} \land \\\quad \logrelC({\sem{\tau}}, {\sem{\alpha}})(\kappa, K) \end{array} \\
    \sem{\Gamma}(\rho, \gamma) & \eqdef {\All (\var \col \tau \in \Gamma). \Box~ \All \Phi. \logrelE({\sem{\tau}, \Phi, \Phi})(\rho(\var), \gamma(\var))} \\
    \Gamma \vDash_{\mathsf{pure}} \expr \col \tau & \eqdef \Box~ \All (\rho, \gamma). \sem{\Gamma}(\rho, \gamma) \wand  \All \Phi. \logrelE({\sem{\tau}, \Phi, \Phi})(\denotE{e}_\rho, \plug{e}{\gamma}) \\
    \Gamma; \alpha \vDash \expr \col \tau; \beta & \eqdef \Box~ \All (\rho, \gamma). \sem{\Gamma}(\rho, \gamma) \wand  \logrelE({\sem{\tau}, \sem{\alpha}, \sem{\beta}})(\denotE{e}_\rho, \plug{e}{\gamma})
  \end{align*}
  \caption{Logical relation for $\delimlang$.}
  \label{fig:d:delim_logrel}
\end{figure*}
We prove adequacy using a logical relation.
It relates expressions to their interpretations and also connects syntactic and semantic configurations.
The logical relation is shown in \Cref{fig:d:delim_logrel}.
It is again a form of biorthogonal logical relation, with the main focus being the observational refinement $\logrelO$: two configurations are related if they reduce to the same natural number.
This coincides with what we want to show in \Cref{{lem:d:adequacy}}.
To facilitate this we then lift $\logrelO$ to the levels of metacontinuations, continuations and expressions.
The relation $\logrelM(P)$, where $P : \vrel$ is a relation on values, states that two metacontinuations are related if, whenever we plug in $P$-related values, the resulting configurations become $\logrelO$-related.
Both $\logrelM$ and $\logrelO$ are then used to define the relation between semantic and syntactic continuations.
The relation $\logrelC(Q, P)$, where $P, Q : \vrel$ are relations on values, states that two continuations are related if, whenever we plug them into $Q$-related metacontinuations with $P$-related values, the resulting configurations become $\logrelO$-related.
Finally, we use $\logrelC, \logrelM$, and $\logrelO$ to define the relation between \gitrees and $\delimlang$ terms.
The relation $\logrelE(P, Q, R)$, where $P, Q, R : \vrel$ are relations on values, states that $\beta$ is related to $e$ if, whenever we plug them into $(P,Q)$-related continuations and an $R$-related metacontinuation, the resulting configurations become $\logrelO$-related.

The relations $\logrelE$ and $\logrelC$ are used to give semantics $\sem{\tau}$ to types.
The idea is that $\logrelE{(\sem{\tau}, \sem{\alpha}, \sem{\beta})}$ relates terms $\emptyset ; \alpha \vdash e : \tau; \beta$ to their semantic counterparts.
This is then used, as expected for logical relations, for defining the logical relation for function types and for open terms.
The relation $\sem{\Gamma}(\rho, \gamma)$ relates the semantic environment $\rho : \Var \to \IT$ to the syntactic substitution $\gamma : \Var \to \Expr$; they are related if they map the same variables to related \gitrees/expressions.
Then we say that an expression $e$ is semantically valid, $\Gamma ; \alpha \vdash e : \tau; \beta$, if its interpretation $\denotE{\expr}_{\rho}$ is related to $e[\gamma]$ under related substitutions $\rho, \gamma$.
Note that if we ignore the answer types we can see that the logical relation exhibits a lot of similarities to the logical relation we gave in \Cref{{global:sec:callcc-sem}}, and follows the same roadmap.

For this logical relation we obtain the fundamental property, which we will use for the proof of adequacy.
\begin{lemma}{Fundamental lemma.}
  \label{lem:d:fundamental}
  Let $\typeD{\Gamma}{\expr}{\tau}{\alpha}{\beta}$
  then $\Gamma; \alpha \vDash \expr \col \tau; \beta$;
  and if $\typeDpure{\Gamma}{\expr}{\tau}$
  then $\Gamma \vDash_{\mathsf{pure}} \expr \col \tau$.
\end{lemma}
\begin{proof}[of \Cref{lem:d:adequacy}]
  Note that the empty (meta)continuation is related to its denotation:
  $\logrelC(P, P)(\idfun, \emptyK)$ and $\logrelM(P)([], [])$ hold for any relation $P$.

  With this, we instantiate $\emptyset; \mathbb{N} \models e : \mathbb{N}; \mathbb{N}$ (that we get from \Cref{lem:d:fundamental})  with the empty continuation/metacontinuation, and get the observational refinement between $e$ and $\denotE{e}$. \qed
\end{proof}

This completes our treatment of denotational semantics of $\delimlang$.
The next section examines interoperability of delimited continuations and other effects.

 }

\section{Modeling Interoperability Between Languages}
\label{global:sec:ffi}
A key advantage of using (G)Itrees for semantics is that they can provide a common framework for multi-language interaction.
This section presents a case study on the interaction between the languages $\embedstlang$ (with higher-order store effects) and $\delimlang$ (with delimited continuations).
Specifically, we allow embedding closed $\delimlang$ programs into $\embedstlang$, and equip $\embedstlang$ with a type system that guarantees safe interoperability.

The embedding we provide is restrictive, preventing programs with delimited continuations from accessing outer-language continuations.
We leave developing a more permissive type system for future work.
At the end of the section we give an example of how to verify a more involved interaction of effects, albeit without the type system.

In this section we reuse the semantics of $\delimlang$ from the previous section and higher-order store effects from \Cref{global:sec:gitrees}.
For $\delimlang$, we reuse the semantics from the previous section and higher-order store reifiers from~\cite{DBLP:journals/pacmpl/FruminTB24} In this section, we use \textcolor{magenta}{magenta} to explicitly highlight programs written in $\delimlang$, and for the interpretation functions of the denotational model of $\delimlang$.

\paragraph{Language $\embedstlang$.}
\begin{figure}[t]
  \begin{mathpar}
      {\begin{grammar}
      \text{types} & \Tys \ni \tau & \tnat \mid \Tunit \mid \tarr{\tau}{\sigma} \mid \Tref(\tau) \\
      \text{expressions} & \Expr \ni \expr & \var \mid () \mid \eapp{\expr_1}{\expr_2} \mid \natop{\expr_1}{\expr_2} \mid  n \mid \Lam \var. \expr \\
      \GrmContinue & \mid \loc \mid \Alloc \expr \mid \deref \expr \mid \expr_1 \gets \expr_2 \mid \embed{} \begingroup\color{magenta} e \endgroup
    \end{grammar}}
    \and
    \inferrule{\typeDpure{\emptyset}{\begingroup\color{magenta} e \endgroup}{\tnat}}{\typeA{\Gamma}{\embed{} \begingroup\color{magenta} e \endgroup}{\tnat}}
    \end{mathpar}
    \caption{Syntax and the new typing rule of the $\embedstlang$.}
    \label{fig:ffi:syn_top_level}
\end{figure}
$\embedstlang$ is a $\lambda$-calculus with base types $\tnat$ and $\Tunit$, references types $\Tref(\tau)$ and
function types $\tarr{\tau}{\sigma}$, with syntax given in \Cref{fig:ffi:syn_top_level}.
Additionally, it includes a construct $\embed{e}$ for embedding $\delimlang$ programs.
The typing rules are all standard, except for the new typing rule for the embedding.

The idea behind the rule is that we can embed an expression from $\delimlang$ if it is a ``pure'' expression that can evaluate to a natural number.
The use of pure typing judgment for the embedded program ensures that it does not to alter the answer type.
This means that we can treat an embedded expression as a ``complete'' program, that does not require outer continuation delimiters, even though it may rely on delimited continuations internally.
Those restrictions are crucial for the type safety of the embedding.
The typing guarantees that $\expr$ does not expect any additional delimiters, but it does not, by itself guarantee that any continuations in $\expr$ escape the embedding boundary.
To prevent that we enforce the continuation delimiter along the embedding boundary in the interpretation of embedded expressions.

\paragraph{Denotational model of $\embedstlang$.}
For denotational semantics of $\embedstlang$, we start by defining reifiers for the effect signature, which includes higher-order store operations (allocating, reading, and storing references) as $E_{\texttt{state}}$, and effects related to delimited continuations ($E_{\delimlang}$).
Then the combined effect signature is $E_{\delimlang} \times E_{\texttt{state}}$, and thus we also let
$\State \eqdef \State_{\texttt{delim}} \times \State_{\texttt{state}}$, and reifiers are defined component-wise.
\begin{figure}[t]
\begin{minipage}{0.45\textwidth}
    \mbox{\fbox{$\denotE{-} \col \Expr \to (\Var \rightarrow \IT) \rightarrow \IT$}}
    \begin{align*}
      &\denotE{\embed{} \begingroup \color{magenta} e \endgroup}_{\rho} = \RESET (\Next (\denotEext{\expr}_\emptyset)) \\
      &\denotE{\var}_\rho = \rho(\var) \\
      &\denotE{\loc}_\rho = \Rret(\loc)
    \end{align*}
  \end{minipage}
  \hfill
  \hspace{-1cm}
  \begin{minipage}{0.45\textwidth}
    \begin{align*}
      &\denotE{\Alloc \expr}_\rho = \getval(\denotE{\expr}_\rho, \lambda x. \ALLOC (x, \Rret))\\
      & \denotE{\deref \expr}_\rho = \getval(\denotE{\expr}_\rho, \lambda x. \READ(x)) \\
  & \denotE{\expr_1 \gets \expr_2}_\rho = \begin{array}[t]{l}
                                                \getval(\denotE{\expr_2}_\rho, \lambda x. \\ \;
                                                \getret(\denotE{\expr_1}_\rho, \lambda y. \WRITE(y, x)))
                                              \end{array}
    \end{align*}
  \end{minipage}
  \begin{minipage}{0.43\textwidth}
\begin{align*}
        \vrel &\eqdef {\ITv} \to \iProp \\
        \erel &\eqdef {\IT} \to \iProp \\
      \logrelO &\col \vrel \to \erel \\
      \logrelO(P)(\beta) &\eqdef \clwpre{\beta}{x \synds P~x \ast
\hasstate([])} \\
      \logrelE &\col \vrel \to \erel \\
        \logrelE({P})({\beta}) &\eqdef \begin{array}[t]{l} \heapctx \wand \hasstate([]) \wand \\
                                         \quad \logrelO(P)(\beta)
                                       \end{array}
    \end{align*}
  \end{minipage}
  \hfill
  \begin{minipage}{0.43\textwidth}
\begin{align*}
      \sem{\tau} &\col \vrel\\
      \sem{\Tunit}(\beta) & \eqdef \beta = () \\
      \sem{\tnat}(\beta) & \eqdef \Exists n. \beta = \Rret(n) \\
      \sem{\tarr{\tau}{\sigma}}(\beta) & \eqdef
                                         \begin{array}[t]{l}
                                           \Exists F. \beta = \Fun(F)~\land \\
                                           \quad \Box~\All \beta. \sem{\tau}(\beta) \wand \\                                                                                                                                 \quad \logrelE(\sem{\sigma})(\APPs{\Fun(F)}{\beta})
                                         \end{array} \\
      \sem{\Tref(\tau)}(\beta) & \eqdef \begin{array}[t]{l}
                                                                           \Exists \ell. \beta = \Rret(\ell)~\land \\\quad \knowInv{}{\Exists \nu. \ell \mapsto \nu \ast \sem{\tau}(\nu)}
                                                                         \end{array}
    \end{align*}
  \end{minipage}
\begin{align*}
\sem{\Gamma}(\rho) & \eqdef {\All (\var \col \tau \in \Gamma). \Box~\logrelE({\sem{\tau}})(\eapp{\gamma}{\var})} &
    \Gamma \vDash \expr \col \tau & \eqdef \All \rho. \sem{\Gamma}(\rho) \wand \logrelE({\sem{\tau}})(\denotE{e}_\rho)
  \end{align*}
\caption{Denotational semantics (selected clauses) and logical relation for $\embedstlang$.}
\label{fig:ffi:semantic}
\end{figure}
~\Cref{fig:ffi:semantic} shows the key parts of the denotational semantics.
For most of the syntactic constructs we give the standard interpretation.
For $\embed{\expr}$ we use the interpretation $\denotEext{-}$ for $\delimlang$ from \Cref{global:sec:delim-sem}, and explicitly wrap the resulting \gitree in a $\RESET$.
This continuation delimiter acts as a sort of \emph{glue code} to protect the rest of the program from being captured by control operators from the embedded $\delimlang$ program.

To show type safety of $\embedstlang$, we construct a logical relation (shown in \Cref{fig:ffi:semantic}), which is similar to the other logical relations we considered in this paper, mainly different in the observation relation $\logrelO$.
Given that the type system for $\embedstlang$ effectively prevents
expressions of $\delimlang$ to access contexts from $\embedstlang$, we
refine the observation relation to get access to a version of the~\ruleref{wp-hom} rule
for expression interpretations of well-typed programs of
$\embedstlang$, which we do not have in general,  as discussed in~\Cref{global:sec:callcc-lang}.

Instead of the standard weakest precondition $\textlog{wp}$, we utilize a \emph{context-local weakest precondition} $\textlog{clwpre}$ which bakes-in the bind rules~\cite{DBLP:journals/pacmpl/TimanyB19}.
\begin{definition}
  Context-local weakest precondition ($\textlog{clwp}$) is defined as follows: $
    \clwpre{\alpha}{\Phi} \triangleq \All \kappa\; (\Psi \col \ITv \to \iProp). (\All v. \Phi~v \wand \wpre{(\eapp{\kappa}{v})}{\Psi}) \to \wpre{(\eapp{\kappa}{\alpha})}{\Psi}
  $
  \label{clwp-def}
\end{definition}
Note that $\textlog{clwp}$ always implies $\textlog{wp}$ and validates a form of the bind rule:
$$\clwpre{\alpha}{\beta \synds \clwpre{(\eapp{\kappa}{\beta})}{\Phi}} \vdash \clwpre{(\eapp{\kappa}{\alpha})}{\Phi}$$
for any homomorphism $\kappa$.
We use the $\textlog{clwp}$ in the definition of observational refinement in the model.
Observational refinement asserts that if the evaluation of a top-level expression of $\embedstlang$ starts with an empty continuation stack, then the evaluation does not introduce new elements into the continuation stack.
By using $\textlog{clwp}$ we get a semantical bind lemma, which can be seen as a version of \ruleref{wp-hom} for semantically valid expressions of $\embedstlang$. As before, we then obtain fundamental lemma and denotational type soundness.
\begin{lemma}{Semantical bind.}
  \label{lem:app:sem-bind}
  $\logrelE(\Lam x \col \ITv. \logrelE(P)(\eapp{\kappa}{x}))(\beta)$ implies
  $\logrelE(P)(\eapp{\kappa}{\beta})$ for any homomorphism $\kappa$.
\end{lemma}
\begin{lemma}{Fundamental lemma.}
  \label{lm:ffi:fund}
  Let $\typeA{\Gamma}{e}{\tau}$ then $\Gamma \vDash e \col \tau$.
\end{lemma}

\begin{lemma}{Denotational type soundness.}
  Let $\typeA{\emptyset}{e}{\tau}$ and $(\denotE{e}_{\emptyset}, []) \istep^* (\alpha, \sigma)$, then $(\Exists \beta\; \sigma'. (\alpha, \sigma) \istep (\beta, \sigma')) \lor (\alpha \in \ITv)$.
  \label{lm:ffi:sound}
\end{lemma}

\paragraph{Unrestricted interaction of delimited continuations and higher-order state.}
Even though the type system we considered here is restrictive, we can still reason about unrestricted interactions of events in the ``untyped'' setting.
Here we show an example of such an unrestricted interaction, and demonstrate how to
reason about context-dependent and context-independent effects at the same time.
While this kind of interactions is forbidden by our type system, we can still write and prove meaningful specifications for such programs.
\begin{figure}[t]
  \begin{multicols}{2}
    \iblock{
      \mhang{
        \texttt{prog} \eqdef \Fun(\Next (\lambda y.
      }{
        \mrow{
          \LET x = \ALLOC (\Rret(1)) IN
        }
        \mhang{
          \mathsf{Let} \spac n = \SHIFT (\lambda k. \Next(
        }{
          \mrow{
            \APPCONT'(\READ(x), k) \SEQ
          }
          \mrow{
            \LET m = \NATOP_+(\READ(x), \Rret(1)) IN \WRITE(x, m) \SEQ
          }
          \mrow{
            \APPCONT'(\READ(x), k)))
          }
          \mrow{
            \mathsf{in} \spac \getret(y, (\lambda l. \LET p = \NATOP_+ (\READ(l), n) IN \WRITE(l, p)))))
          }}}
    }
    \columnbreak
    \iblock{
      \row
      \mrow{\text{Initial offset value}}
      \mrow{\text{Capture continuation as $k$}}
      \mrow{\text{First call to $k$ with the init. value of $x$}}
      \mrow{x := x + 1;}
      \row
      \mrow{\text{Second call to $k$ with updated $x$}}
      \row
      \mrow{y := y + n;}
    }
  \end{multicols}
  \vspace{-0.5cm}
  \[
    \mbox{where } \APPCONT'(x, y) \eqdef \APPCONT(\Next(x), \Next(\Tau \circ y \circ \Next))
  \]
  \begin{mathpar}
  \inferrule{\heapctx \\ \hasstate(\sigma) \\ y \mapsto \Rret(n)}
  {\wpre{(\RESET(\Next (\APPs{\texttt{prog}}{\Rret(y)})))}{y \mapsto \Rret(n + 3) \ast \hasstate (\sigma)}}
  \end{mathpar}
  \caption{Example program with delimited continuations and state and its specification.}
  \label{fig:ffi:prog}
\end{figure}
Consider the program in \Cref{fig:ffi:prog}, written in GITrees directly.
The function \texttt{prog} utilizes both delimited continuations and state.
It takes a reference $y$ as its
argument and begins by allocating the value $1$ in the store at
reference $x$.
Then it captures the continuation
$
\getret(y, (\lambda l. \LET p = \NATOP_+ (\READ(l), \dots) IN \WRITE(l, p)))
$
as $k$.
Invoking the continuation $k$ with a number $n$ increments the current value of $y$ by $n$.
The program then invokes this continuation twice.
First with the original value of $x$.
Then, with an incremented value of $x$.
Since the starting value of $x$ is 1, the reference $y$ is
incremented first by $1$ and then by $2$.
We capture this behavior in the specification for \texttt{prog} stated in \Cref{fig:ffi:prog}.

It is important to note that this program features a bidirectional interaction between state and continuations.
Specifically, the body of $\SHIFT$ involves state operations, while the result of $\SHIFT$ is subsequently used to increment a value in the heap.
As we have seen, while this type of interaction is not allowed in the type system, we can still reason about them in program logic.
We stipulate that our proposed type system could potentially be extended to support embedding ``pure'' functions, allowing for bi-directional interaction between the two languages.
We believe that such an extension would require implementing answer-type polymorphism, following the approach of Asai and Kameyama~\cite{asai:2007}.

\section{Preemptive concurrency}
\label{global:concurrency}
In this section we present yet another extension to \gitrees: support for preemptive concurrency.
This extension is orthogonal to the extension for context-dependent effects presented earlier, and demonstrates the flexibility of the \gitrees framework.
While concurrency does not involve context-dependent reification, it does require modifications to the step relation and the weakest precondition logic.

The key technical challenge when extending \gitrees with concurrency is handling shared mutable state.
In a sequential setting, the state effects $\READ$ and $\WRITE$ can be used independently.
However, in a concurrent setting, separate read and write operations are no longer sufficient: race conditions can occur between a read and a subsequent write.
To address this, we introduce an atomic exchange operation $\XCHG$ that atomically swaps the contents of a location with a new value.
Crucially, the implementation of $\XCHG$ exploits a key feature of \gitrees: effects can take \emph{functions} as inputs.
This allows us to define a generic atomic state modification effect, from which $\XCHG$ and other atomic operations can be derived.

We demonstrate this extension by extending the affine language $\afflang$ of~\citet{DBLP:journals/pacmpl/FruminTB24} with concurrency primitives.
The affine type system ensures that values are used exactly once, which interacts in an interesting way with shared state and concurrency.
In particular, variables in the concurrent setting must be accessed using atomic exchange rather than sequential read/write operations, to maintain the affine discipline.

\subsection{Threadpool Semantics}

To support concurrency, we lift the step relation $\istep$ from individual \gitrees to \emph{threadpools}.
A threadpool is simply a list of \gitrees, representing threads running in parallel.
The step relation for threadpools is defined as:
\begin{align*}
  (\vec{\alpha}, \sigma) \istep_{\mathsf{tp}} (\vec{\alpha}', \sigma') \eqdef{}&
  \Exists i\, \alpha_i\, \alpha_i'. \vec{\alpha}[i] = \alpha_i \wedge (\alpha_i, \sigma) \istep (\alpha_i', \sigma', \vec{l}) \wedge \vec{\alpha}' = \vec{\alpha}[i \mapsto \alpha_i'] + \vec{l}
\end{align*}
That is, one thread in the pool takes a step, updating both the thread and the global state.
The other threads remain unchanged.
Accordingly, we also modify signature of reifiers and allow $\istep$ to take forked threads into the account, as we describe in the next section.

The adequacy theorem (\Cref{thm:wp_adequacy}) extends naturally to threadpool reductions.

\subsection{Fork Effect and Affine Language Extension}

To spawn new threads we introduce a fork effect whose signature is shown in~\Cref{fig:fork_effect}.
The effect takes as input a \gitree $e$ to be executed in a new thread, and returns unit.
The reifier for fork, denoted $r_{\mathtt{fork}}$, takes the input \gitree $e$, leaves the state $\sigma$ unchanged, returns unit to the continuation $\kappa$, and produces a list containing the new thread $e$.
This list is used to extend the threadpool with the newly spawned thread.

More precisely, the reifier signature is modified to return an optional triple $(\beta, \sigma', \vec{\alpha})$ where $\vec{\alpha}$ is a list of new threads to be added to the threadpool.
For effects other than fork this list is empty.

This modification to the reifier signature is analogous to the extension for context-dependent effects presented earlier in this work.
Just as context-dependent effects required reifiers to take the current continuation as an additional parameter (to allow effects like \texttt{call/cc} to capture and manipulate continuations), concurrency requires reifiers to return a list of new threads.
Both extensions maintain backwards compatibility: effects that do not need these capabilities simply ignore the continuation parameter (for context-independent effects), or return an empty thread list (for non-forking effects).
This demonstrates the flexibility of the \gitrees framework: the reifier mechanism can be extended to accommodate different classes of effects while preserving the modularity of effect composition.

We demonstrate the fork effect by extending the affine language $\afflang$ with a $\aFork$ construct.
The syntax and type system are shown in~\Cref{fig:afflang_typing}.
The expression $\aFork(\expr_1, \expr_2)$ first evaluates $\expr_1$ (which typically performs some side effect), spawns it in a new thread, and then proceeds with $\expr_2$.
The type system ensures that $\expr_1$ has type $\tUnit$, so that no affine resources are lost when forking.

\begin{figure}[t]
  \begin{align*}
    & r \col \prod_{\idx \in E} \Ins_{\idx}(\IT_E) \times \stateO \times (\Outs_{\idx}(\IT_E) \rightarrow \latert \IT_E) \to \optionO(\latert \IT_E \times \stateO \times \mathsf{list} (\latert \IT_E))
  \end{align*}
  \begin{align*}
    \Ins_{\mathtt{fork}}(X) & \eqdef \latert X & \Outs_{\mathtt{fork}}(X) & \eqdef 1
  \end{align*}
  \begin{align*}
    r_{\mathtt{fork}}(e, \sigma, \kappa) &= \Some(\kappa~(), \sigma, [e]) \\
    \FORK(e) & \eqdef \Vis_{\mathtt{fork}}(\Next(e), \lambda x. \Next(\Rret(x)))
  \end{align*}
  \caption{Modified reifier signature and fork effect.}
  \label{fig:fork_effect}
\end{figure}

\begin{align*}
\typeAff \in \langkwa{\Type} \bnfdef{}&
\tBool \ALT \tNat \ALT \tUnit
\ALT \typeAff_1 \atensor \typeAff_2 \ALT \typeAff_1 \lolli \typeAff_2
\ALT \tRef{\typeAff}
\\
\expr \in \langkwa{\Expr} \bnfdef{}&
  n \ALT b \ALT \unittt \ALT
  \avar \ALT
\LamA \avar.\expr \ALT
\expr_1\ \expr_2 \ALT
\aPair{\expr_1}{\expr_2} \ALT
\LetA \aPair{\avar_1}{\avar_2}=\expr_1 in \expr_2 \\& \ALT \aAlloc(\expr) \ALT \aDealloc(\expr_1) \ALT
\aReplace(\expr_1,\expr_2) \ALT \aFork(\expr_1, \expr_2)
\end{align*}

\begin{figure}[t]\small
  \begin{mathpar}
    \infer
    {b \in \mathbb{B}}
    {\aCtx \provesAff b : \tBool}
    \and
    \infer
    {n \in \mathbb{N}}
    {\aCtx \provesAff n : \tNat}
    \and
    \axiom
    {\aCtx \provesAff \unittt : \tUnit}
    \and
    \axiom
    {\aCtx_1, \avar : \typeAff, \aCtx_2 \provesAff \avar : \typeAff}
    \and
    \infer
    {\avar : \typeAff_1, \aCtx \provesAff \expr : \typeAff_2}
    {\aCtx \provesAff \LamA \avar. \expr : \typeAff_1 \lolli \typeAff_2}
    \and
    \infer
    {\aCtx_1 \provesAff \expr_1 : \typeAff_1 \lolli \typeAff_2
      \and \aCtx_2 \provesAff  \expr_2 : \typeAff_1
    }
    {\aCtx_1, \aCtx_2 \provesAff \expr_1 \ \expr_2 : \typeAff_2}
    \and
    \infer
    {\aCtx_1 \provesAff \expr_1 : \typeAff_1
      \and \aCtx_2 \provesAff  \expr_2 : \typeAff_2}
    {\aCtx_1,\aCtx_2 \provesAff \aPair{\expr_1}{\expr_2} :  \typeAff_1 \atensor \typeAff_2}
    \and
    \infer
    {\aCtx_1 \provesAff \expr_1 : \typeAff_1 \atensor \typeAff_2
      \and \avar_1:\typeAff_1,\avar_2:\typeAff_2,\aCtx_2\provesAff \expr_2 : \typeAff
    }
    {\aCtx_1, \aCtx_2 \provesAff \LetA \aPair{\avar_1}{\avar_2}=\expr_1 in \expr_2 : \typeAff}
    \and
    \infer
    {\aCtx \provesAff \expr : \typeAff}
    {\aCtx \provesAff \aAlloc(\expr) : \tRef{\typeAff}}
    \and
    \infer
    {\aCtx \provesAff \expr : \tRef{\typeAff}}
    {\aCtx \provesAff \aDealloc(\expr) : \tUnit}
    \and
    \infer
    {\aCtx_1 \provesAff \expr_1 : \tRef{\typeAff_1}
      \and \aCtx_2 \provesAff \expr_2 : \typeAff_2}
    {\aCtx_1,\aCtx_2  \provesAff \aReplace(\expr_1,\expr_2) : \typeAff_1 \atensor\, \tRef{\typeAff_2}}
    \and
    \infer
    {\aCtx_1 \provesAff \expr_1 : \tUnit
    \and \aCtx_2 \provesAff \expr_2 : \typeAff}
    {\aCtx_1, \aCtx_2 \provesAff \aFork(\expr_1, \expr_2) : \typeAff}
  \end{mathpar}
  \caption{Type system for $\afflang$ with concurrency.}
  \label{fig:afflang_typing}
\end{figure}

\subsection{Atomic State Modification via Lambda Functions in Effects}

The interpretation of $\afflang$ shown in~\Cref{fig:afflang_interp} shows an important aspect of concurrent programming with affine types.
Consider the interpretation of variables: $\Sem{\avar}_\rho \eqdef \Force(\rho(\avar))$, where:
\begin{align*}
 \Thunk(\alpha) \eqdef{}& \ALLOC\big(\Rret(0),\Lam \ell.
    \Fun(\Next(\Lam \_. \IF(\XCHG(\ell, \Rret(1)), \Err(Lin), \alpha)))\big)\\
  \Force(\alpha) \eqdef{}& \APPs{\alpha}{\Rret(0)}
\end{align*}
The $\Thunk$ operation allocates a reference initialized to $0$ (representing ``not yet used''), and returns a thunk that, when forced, atomically exchanges the location's value with $1$.
If the exchange returns $0$, the variable has not been used before, and we proceed with $\alpha$.
If it returns $1$, the variable has already been used, violating the affine discipline, and we raise a linearity error.

\paragraph{Generic Atomic State Modification.}
The key insight enabling $\XCHG$ is that \gitrees allow effects to take \emph{lambda functions} as inputs.
We define a generic atomic state modification effect $\mathtt{atomic}$ as follows:
\begin{align*}
  \Ins_{\mathtt{atomic}}(X) &\eqdef \Loc \times \latert (X \to X \times X) \\
  \Outs_{\mathtt{atomic}}(X) &\eqdef \latert X \\ 
  r_{\mathtt{atomic}}((\ell, f), \sigma, \kappa) &= \begin{cases}
    \Some(\kappa~(\latertinj(r)), \sigma[\ell \mapsto \latertinj(v')]), & \sigma[\ell] = \Some(v), f(v) = \latertinj(r, v') \\ 
    \None, & \text{otherwise}
  \end{cases}
\end{align*}
This effect takes a location $\ell$ and a function $f : \latert (X \to X \times X)$.
The reifier atomically reads the current value $v = \sigma(\ell)$, applies $f$ to it to obtain a pair $(r, v')$, stores $v'$ back in $\ell$, and returns $r$ to the continuation.
This design allows the function to compute both what to return and what to write, enabling a wide variety of atomic operations.

From this generic effect, we can derive various atomic operations.
Atomic exchange is defined as:
\begin{align*}
  \XCHG(\ell, w) &\eqdef \Vis_{\mathtt{atomic}}((\ell, \Lam x. (\Next(x), \Next(w))), \idfun)
\end{align*}
The function $\Lam x. (\Next(x), \Next(w))$ returns the old value $x$ as the result and the new value $w$ to be written.
Compare-and-swap (CAS) can be defined as:
\begin{align*}
  \mathtt{CAS}(\ell, \mathit{exp}, \mathit{new}) &\eqdef \Vis_{\mathtt{atomic}}((\ell, \Lam x. \IF(x = \mathit{exp}, (\Rret(\mathit{true}), \Next(\mathit{new})), (\Rret(\mathit{false}), \Next(x)))), \idfun)
\end{align*}
which returns true and updates if the expected value matches, otherwise returns false and leaves the value unchanged.
Fetch-and-add (FAA) is defined as:
\begin{align*}
  \mathtt{FAA}(\ell, n) &\eqdef \Vis_{\mathtt{atomic}}((\ell, \Lam x. (\Next(x), \Next(\NATOP_+(\Next(x), \Next(n))))), \idfun)
\end{align*}
which returns the old value and stores the incremented value.

\paragraph{Reasoning about Atomic Operations.}
The program logic rule for the atomic effect captures its semantics of both reading from, and writing to the state at the same time:
\begin{mathpar}
  \inferH{wp-atomic}
  {
    \begin{array}[c]{c}
      \heapctx \arcr
      \later \ell \mapsto \alpha \arcr
      \later (f~\alpha \equiv \Next(w_1, w_2)) \arcr
      \later (\ell \mapsto w_2 \wand {\Phi(w_1)})
    \end{array}
  }
  {\wpre{\Vis_{\mathtt{atomic}}((\ell, f), \idfun)}{\Phi}}
\end{mathpar}
Given ownership of location $\ell$ pointing to value $\alpha$ (written $\ell \mapsto \alpha$), and knowing that applying function $f$ to $\alpha$ yields the pair $(w_1, w_2)$, the rule allows us to conclude that after executing the atomic operation, we regain ownership of $\ell$ now pointing to $w_2$, and the result $w_1$ satisfies the postcondition.
This rule is atomic in the sense that the read-modify-write sequence happens in a single logical step, which is essential for reasoning about concurrent programs where intermediate states must not be observable by other threads.

This demonstrates a unique capability of \gitrees: by allowing effects to take functions as inputs, we can express atomic state modification generically, and derive various atomic primitives as needed.
This is more flexible than hardcoding specific atomic operations into the effect signature.

\subsection{Interpretation of $\afflang$ with Concurrency}

The interpretation of $\afflang$ expressions is shown in~\Cref{fig:afflang_interp}.
Most clauses are standard, but there are two important differences from the sequential version:
First, the interpretation of $\aReplace(\expr_1, \expr_2)$ now uses $\XCHG$ instead of sequential $\READ$ and $\WRITE$ operations, reflecting the new implementation of $\Thunk$.
Second, the interpretation of $\aFork(\expr_1, \expr_2)$ uses the $\FORK$ effect to spawn $\expr_1$ in a new thread before proceeding with $\expr_2$.

\begin{figure}[t]
  \begin{align*}
    \Sem{b}_\rho \eqdef{}&
                           \begin{cases}
                             \Rret(1)&\mbox{if b = $\langkwa{true}$}\\
                             \Rret(0)&\mbox{otherwise}
                           \end{cases}
    & \Sem{\unittt}_\rho \eqdef{} & \Rret()\\
    \Sem{n} \eqdef&{} \Rret(n)
                                    &  \Sem{\LamA \avar.\expr}_\rho  \eqdef{}&
                                                                               \Fun(\Next(\Lam \alpha.\Sem{\expr}_{\rho[\avar \mapsto \alpha]})) \\
    \Sem{\avar}_\rho \eqdef{}& \Force(\rho(\avar))
                                    &  \Sem{\expr_1\ \expr_2}_\rho \eqdef{}&
                                                                             \begin{aligned}[t]
                                                                               &\LET x = \Sem{\expr_2}_{\rho_2} IN \\
                                                                               &\APPs{\Sem{\expr_1}_{\rho_1}}{\Thunk(x)}
                                                                             \end{aligned} \\
    \Sem{\LetA \aPair{\avar_1}{\avar_2}=\expr_1 in \expr_2}_{\rho} \eqdef{}&
                                                                             \begin{aligned}[t]
                                                                               &\LET x=\Sem{\expr_1}_{\rho_1} IN\\
                                                                               &\LET y=\Thunk(\Proj{1}(x)) IN\\
                                                                               &\LET z=\Thunk(\Proj{2}(x)) IN\\
                                                                               &\Sem{\expr_2}_{\rho_2[\avar_1\mapsto y,\avar_2\mapsto z]}
                                                                             \end{aligned}
    &  \Sem{\aPair{\expr_1}{\expr_2}}_{\rho} \eqdef{}&
                                                       (\Sem{\expr_1}_{\rho_2},\Sem{\expr_2}_{\rho_2}) \\
    \Sem{\aAlloc(\expr)}_{\rho} \eqdef{}&
                                          \LET x=\Sem{\expr}_\rho IN \ALLOC(x,\Rret)
                                    &  \Sem{\aDealloc(\expr)}_\rho \eqdef{}&
                                                                             \begin{aligned}[t]
                                                                               &\getret(\Sem{\expr}_\rho,\DEALLOC)
                                                                             \end{aligned}\\
    \Sem{\aReplace(\expr_1,\expr_2)}_{\rho} \eqdef{}& \multispan{3}{$\begin{aligned}[t] &\LET y = \Sem{\expr_2}_{\rho_2} IN \getret(\Sem{\expr_1}_{\rho_1}, \Lam \ell.\LET x = \XCHG(\ell,y) IN\\ & \quad (x, \Rret(\ell)))\end{aligned}$}\\
    \Sem{\aFork(\expr_1, \expr_2)}_{\rho} \eqdef{}&\FORK(\Sem{\expr_1}_{\rho_1}) \SEQ \Sem{\expr_2}_{\rho_2}
  \end{align*}

  \caption{Interpretation of $\afflang$ with concurrency.}
  \label{fig:afflang_interp}
\end{figure}

\begin{figure}[t]\small
  \begin{align*}
    \AVV{\tUnit}(\beta_v) \eqdef{}& \beta_v = \Rret()
&  \AVV{\tNat}(\beta_v) \eqdef{}& \Exists n \in \mathbb{N}. \beta_v = \Rret(n)
\\
  \AVV{\type_1 \lolli \type_2}(\beta_v) \eqdef{}&
                                                     \All \alpha_w. \AVV{\type_1}(\alpha_w)
                                                     \wand \AEE{\type_2}(\APPs{\beta_v}{\alpha_w})
& \AVV{\tBool}(\beta_v) \eqdef{}& \beta_v = \Rret(0) \vee \beta_v = \Rret(1)
    \\ \AVV{\tRef{\type}}(\beta_v) \eqdef{}&
      \begin{aligned}[t]
      &\Exists \ell \in \Loc,\, \alpha_v. \big(\beta_v = \Rret(\ell))\big) \ast{}\\
      & \ell\mapsto \alpha_v \ast \AVV{\type}(\alpha_v)
      \end{aligned}
& \AVV{\type_1 \atensor \type_2}(\beta_v) \eqdef{}&
  \begin{aligned}[t]
    &\Exists \gamma_v,\delta_v. \beta_v = (\gamma_v,\delta_v) \ast{}\\
    &\AVV{\type_1}(\gamma_v) \ast \AVV{\type_2}(\delta_v)
  \end{aligned}
\\\AEE{\type}(\alpha) \eqdef{}&
      \heapctx \wand \wpre{\alpha}{\Ret \beta_v. \AVV{\type}(\beta_v)}
& \protec(\Phi)(\beta_v) \eqdef{}& \wpre{\Force(\beta_v)}{\Phi}
\\ \AVVCtx{\aCtx}(\rho) \eqdef{}&\All (\avar :\type) \in \aCtx.
   \protec(\AVV{\type})(\rho(\avar))
&  \aCtx \modelsAff \alpha: \typeAff \eqdef{}
      \All \rho.{}& \AVVCtx{\aCtx}(\rho) \implies
          \AEE{\type}(\alpha(\rho))
  \end{align*}
\caption{Logical relation for $\afflang$ with concurrency.}
\label{fig:afflang_logrel}
\end{figure}

\subsection{Program Logic and Type Safety}

To reason about concurrent programs using \gitrees, we extend the weakest precondition logic with rules for the new effects.
The key insight is that concurrent effects require \emph{invariants} to reason about shared state.

We define two invariants:
\begin{align*}
\mathsf{heapctx} &\eqdef \knowInv{\namesp_1}{\Exists \sigma. \hasstate_{\mathtt{heap}}(\sigma) \ast \ownGhost{}{\authfull \sigma}} \\
\mathsf{forkctx} &\eqdef \knowInv{\namesp_2}{\hasstate_{\mathtt{fork}}(())}
\end{align*}
These invariants ensure that the state is always owned by exactly one thread at a time, and can be temporarily acquired when performing an effectful operation.

The weakest precondition rules for the concurrent effects are shown below.
The rule for $\XCHG$ requires ownership of the location $\ell \mapsto \alpha$, and after the exchange, the ownership is updated to $\ell \mapsto w$.
Crucially, this rule is atomic: the read and write happen in a single step, preventing race conditions.

The rule for $\FORK$ requires the invariant $\mathsf{forkctx}$ and ensures that the forked thread $\alpha$ is safe (i.e., $\wpre{\alpha}{\top}$).
After forking, the parent thread continues with $\Phi(\Rret())$.

\begin{mathpar}\small
\infer
  {\mathsf{heapctx} \and
    \later \ell \mapsto \alpha
    \and
    \later (\ell \mapsto w \wand \wpre{\alpha}{\Phi})}
  {\wpre{(\XCHG(\ell, w))}{\Phi}}
  \and
  \infer
  {\mathsf{forkctx} \and
    \later \wpre{\alpha}{\top}
    \and
    \later \Phi(\Rret())}
  {\wpre{(\FORK(\alpha))}{\Phi}}
\end{mathpar}

Using these rules, we verify the interpretation of $\afflang$ shown in~\Cref{fig:afflang_interp}.
The logical relation in~\Cref{fig:afflang_logrel} extends the standard logical relation for affine types with support for references and concurrency.
In particular, the interpretation of reference types $\AVV{\tRef{\type}}(\beta_v)$ requires that the value is a location $\ell$ pointing to a value of type $\type$.

The semantic typing judgment $\aCtx \modelsAff \alpha: \typeAff$ states that the interpretation $\alpha$ of an expression is safe and produces a value of type $\typeAff$ when executed in an environment satisfying $\aCtx$.
Using this logical relation, we prove type soundness for $\afflang$ with concurrency:

\begin{lemma}[Fundamental Lemma for $\afflang$]
  If $\aCtx \provesAff \expr : \typeAff$, then $\aCtx \modelsAff \Sem{\expr} : \typeAff$.
\end{lemma}

The type safety theorem for concurrent \gitrees ensures that well-typed programs are safe and make progress:

\begin{theorem}[Type Safety for Concurrent $\afflang$]
  Let $\provesAff \expr : \typeAff$. Then for any threadpool $\vec{\alpha}$ and state $\sigma$ such that $(\Sem{\expr}, \sigma) \istep_{\mathsf{tp}}^{\ast} (\vec{\alpha},\sigma')$, all threads in $\vec{\alpha}$ are either values or can take a step.
\end{theorem}

This demonstrates that the concurrency extension to \gitrees is sound and can be used to verify language interactions of a language with concurrency, shared mutable state and affine types with other languages.

\section{Discussion and Related Work}
\label{global:sec:conc}
We conclude the paper by discussing related work and future directions.

This paper extends GITrees to handle context-dependent effects, which allows us to model higher-order languages with control operators like $\callcc{\var}{\expr}$ and $\control{\var}{\expr}$.
We provided step-by-step guidance on designing such effects through the example of exceptions.
Additionally, we presented an orthogonal extension for preemptive concurrency, demonstrating how higher-order effects enable generic atomic operations.
We showed this extension supports interoperability between languages with different context-dependent effects, while preserving reasoning about context-independent effects.
Our approach leverages the native support for higher-order functions and effects in \gitrees.
This differs from the first-order effect representation of effects in ITrees~\cite{DBLP:journals/pacmpl/XiaZHHMPZ20}, which would require explicitly first-order representation of functions and continuations, if we want to model first-class continuation.
Such model would mix syntactic and semantic concerns, which is part of what we are trying to avoid by working with (G)ITrees.

Another difference with ITrees is our approach to reasoning.
While ITrees use bisimulation-based equational theory, we follow GITrees in using tailored program logics and defining refinements.
Our logic are expressive enough to define logical relations and carry out computational adequacy proofs.
In future work, it would be interesting to develop techniques for reasoning about weaker notions of equality than the basic equational theory that GITrees comes equipped with, see the more extensive discussion of this point in \cite{DBLP:journals/pacmpl/FruminTB24}.

The ``classical'' domain theory remains an important source of inspiration and ideas for our development, and we want to mention some of the related work along those lines.
Cartwright and Felleisen~\cite{cartwright:1994} introduced a framework of extensible direct models for constructing modular denotational semantics of programming languages.
Their framework centers on an abstract notion of \emph{resources} for representing effects (such as store or continuations) and a central \texttt{admin} function that manages these resources.
Each language extension defines both the types (domains) of additional values and resources, and specifies the \emph{actions} that the \texttt{admin} function can perform on these resources.
Building on this framework, Sitaram and Felleisen~\cite{sitaram:1991} demonstrated that such models can provide direct-style fully abstract semantics for control operators. Their approach interprets effects, including continuations, by delegating them to a top-level handler.
Our work adopts several key ideas from this line of research but reformulates them in the context of \gitrees rather than classical domain theory.
In our framework, effect signatures define resources, reifiers specify actions, and the reduction relation serves as the central authority dispatching the effects.
The transition to \gitrees enables us to formalize the extensibility of this approach in a practical manner, and it allows us to develop program logics where ``resources'' (as above) become resources in separation logic.

Compared to other programming languages paradigms and effects, type systems and logical relations for delimited continuations have not been studied as comprehensively.
The original type system for \texttt{shift}/\texttt{reset} is due to Danvy and Filinski \cite{danvy.filinski.1989}, where they employ \emph{answer-type modification}.
Materzok and Biernacki \cite{materzok:2011} generalized this type system to account for more involved control operators \texttt{shift}$_{0}$/\texttt{reset}$_0$; an alternative substructural type system for these operators was designed by Kiselyov and Shan \cite{kiselyov:2007}.
Dyvbig, Peyton Jones, and Sabry \cite{dyvbig:2007} provide a typed monadic account of CPS for delimited continuations with dynamic prompt generation.
Asai and Kameyama~\cite{asai:2007} present a polymorphic variant of the Danvy and Filinski's type system.

Biernacka and Biernacki~\cite{DBLP:conf/ppdp/BiernackaB09} prove termination for a language with $\control{\var}{\expr}$ and $\delim{\expr}$ (but without recursion) using logical relations based on abstract machine semantics.
The shape of their logical relation is similar to our logical relation
used for showing adequacy in that they also have relations for configurations, metacontinuations, \etc{}
Asai \cite{DBLP:conf/sfp/Asai05} uses type-directed logical relations to verify a direct-style specializer (partial evaluator) for a language with $\control{\var}{\expr}$ and $\delim{\expr}$, proving correctness against evaluation-context based operational semantics.
In contrast with those works, we define our logical relations on denotations, using the semantics of \gitrees and the derived program logics.

In our interoperability example we showed type safety of a combined language, with respect to denotational semantics.
In future work we would like to examine other properties: for example, that $\delimlang$ expressions cannot disrupt $\embedstlang$'s control flow, perhaps establishing some form of well-bracketedness as in \cite{DBLP:journals/pacmpl/TimanyGB24}.

We would also like to study other context-dependent effects like handlers and algebraic effects \cite{PlotkinPretnar:2013,BauerPretnar:2015,WuEtAl:2014,vandenBergEtAl:2021,BachPoulsenvanderRest:2023}.
In particular, it would be interesting to give a denotational semantics to a language with handlers, derive program logic rules for it, and compare the resulting program logic to the one in \cite{vilhena-pottier}.
Additionally, the concurrency extension could be developed further to handle more sophisticated synchronization primitives and scheduling policies.

\begin{acks}
  This work was co-funded by the European Union (ERC, CHORDS, 101096090).
  Views and opinions expressed are however those of the author(s) only and
  do not necessarily reflect those of the European Union or the European
  Research Council. Neither the European Union nor the granting authority
  can be held responsible for them.
  This work was supported in part by a Villum Investigator grant (no.
  25804), Center for Basic Research in Program Verification (CPV), from
  the VILLUM Foundation.
\end{acks}

\bibliographystyle{ACM-Reference-Format}

\begin{thebibliography}{36}



\ifx \showCODEN    \undefined \def \showCODEN     #1{\unskip}     \fi
\ifx \showISBNx    \undefined \def \showISBNx     #1{\unskip}     \fi
\ifx \showISBNxiii \undefined \def \showISBNxiii  #1{\unskip}     \fi
\ifx \showISSN     \undefined \def \showISSN      #1{\unskip}     \fi
\ifx \showLCCN     \undefined \def \showLCCN      #1{\unskip}     \fi
\ifx \shownote     \undefined \def \shownote      #1{#1}          \fi
\ifx \showarticletitle \undefined \def \showarticletitle #1{#1}   \fi
\ifx \showURL      \undefined \def \showURL       {\relax}        \fi
\providecommand\bibfield[2]{#2}
\providecommand\bibinfo[2]{#2}
\providecommand\natexlab[1]{#1}
\providecommand\showeprint[2][]{arXiv:#2}

\bibitem[Asai(2005)]{DBLP:conf/sfp/Asai05}
\bibfield{author}{\bibinfo{person}{Kenichi Asai}.} \bibinfo{year}{2005}\natexlab{}.
\newblock \showarticletitle{Logical relations for call-by-value delimited continuations}. In \bibinfo{booktitle}{\emph{Revised Selected Papers from the Sixth Symposium on Trends in Functional Programming, {TFP} 2005, Tallinn, Estonia, 23-24 September 2005}} \emph{(\bibinfo{series}{Trends in Functional Programming}, Vol.~\bibinfo{volume}{6})}, \bibfield{editor}{\bibinfo{person}{Marko C. J.~D. van Eekelen}} (Ed.). \bibinfo{publisher}{Intellect}, \bibinfo{pages}{63--78}.
\newblock


\bibitem[Asai and Kameyama(2007)]{asai:2007}
\bibfield{author}{\bibinfo{person}{Kenichi Asai} {and} \bibinfo{person}{Yukiyoshi Kameyama}.} \bibinfo{year}{2007}\natexlab{}.
\newblock \showarticletitle{Polymorphic {{Delimited Continuations}}}. In \bibinfo{booktitle}{\emph{Programming {{Languages}} and {{Systems}}}} (Berlin, Heidelberg, 2007), \bibfield{editor}{\bibinfo{person}{Zhong Shao}} (Ed.). \bibinfo{publisher}{Springer}, \bibinfo{pages}{239--254}.
\newblock
\showISBNx{978-3-540-76637-7}
\href{https://doi.org/10.1007/978-3-540-76637-7_16}{doi:\nolinkurl{10.1007/978-3-540-76637-7_16}}


\bibitem[Bach~Poulsen and {van der Rest}(2023)]{BachPoulsenvanderRest:2023}
\bibfield{author}{\bibinfo{person}{Casper Bach~Poulsen} {and} \bibinfo{person}{Cas {van der Rest}}.} \bibinfo{year}{2023}\natexlab{}.
\newblock \showarticletitle{Hefty {{Algebras}}: {{Modular Elaboration}} of {{Higher-Order Algebraic Effects}}}.
\newblock \bibinfo{journal}{\emph{Proceedings of the ACM on Programming Languages}} \bibinfo{volume}{7}, \bibinfo{number}{POPL} (\bibinfo{date}{Jan.} \bibinfo{year}{2023}), \bibinfo{pages}{62:1801--62:1831}.
\newblock
\href{https://doi.org/10.1145/3571255}{doi:\nolinkurl{10.1145/3571255}}


\bibitem[Bauer and Pretnar(2015)]{BauerPretnar:2015}
\bibfield{author}{\bibinfo{person}{Andrej Bauer} {and} \bibinfo{person}{Matija Pretnar}.} \bibinfo{year}{2015}\natexlab{}.
\newblock \showarticletitle{Programming with Algebraic Effects and Handlers}.
\newblock \bibinfo{journal}{\emph{Journal of Logical and Algebraic Methods in Programming}} \bibinfo{volume}{84}, \bibinfo{number}{1} (\bibinfo{date}{Jan.} \bibinfo{year}{2015}), \bibinfo{pages}{108--123}.
\newblock
\showISSN{2352-2208}
\href{https://doi.org/10.1016/j.jlamp.2014.02.001}{doi:\nolinkurl{10.1016/j.jlamp.2014.02.001}}


\bibitem[Biernacka and Biernacki(2009)]{DBLP:conf/ppdp/BiernackaB09}
\bibfield{author}{\bibinfo{person}{Malgorzata Biernacka} {and} \bibinfo{person}{Dariusz Biernacki}.} \bibinfo{year}{2009}\natexlab{}.
\newblock \showarticletitle{Context-based proofs of termination for typed delimited-control operators}. In \bibinfo{booktitle}{\emph{Proceedings of the 11th International {ACM} {SIGPLAN} Conference on Principles and Practice of Declarative Programming, September 7-9, 2009, Coimbra, Portugal}}, \bibfield{editor}{\bibinfo{person}{Ant{\'{o}}nio Porto} {and} \bibinfo{person}{Francisco~Javier L{\'{o}}pez{-}Fraguas}} (Eds.). \bibinfo{publisher}{{ACM}}, \bibinfo{pages}{289--300}.
\newblock
\href{https://doi.org/10.1145/1599410.1599446}{doi:\nolinkurl{10.1145/1599410.1599446}}


\bibitem[Biernacka et~al\mbox{.}(2005)]{DBLP:journals/lmcs/BiernackaBD05}
\bibfield{author}{\bibinfo{person}{Malgorzata Biernacka}, \bibinfo{person}{Dariusz Biernacki}, {and} \bibinfo{person}{Olivier Danvy}.} \bibinfo{year}{2005}\natexlab{}.
\newblock \showarticletitle{An Operational Foundation for Delimited Continuations in the {CPS} Hierarchy}.
\newblock \bibinfo{journal}{\emph{Log. Methods Comput. Sci.}} \bibinfo{volume}{1}, \bibinfo{number}{2} (\bibinfo{year}{2005}).
\newblock
\href{https://doi.org/10.2168/LMCS-1(2:5)2005}{doi:\nolinkurl{10.2168/LMCS-1(2:5)2005}}


\bibitem[Biernacki et~al\mbox{.}(2019)]{10.1145/3290319}
\bibfield{author}{\bibinfo{person}{Dariusz Biernacki}, \bibinfo{person}{Maciej Pir\'{o}g}, \bibinfo{person}{Piotr Polesiuk}, {and} \bibinfo{person}{Filip Sieczkowski}.} \bibinfo{year}{2019}\natexlab{}.
\newblock \showarticletitle{Abstracting algebraic effects}.
\newblock \bibinfo{journal}{\emph{Proc. ACM Program. Lang.}} \bibinfo{volume}{3}, \bibinfo{number}{POPL}, Article \bibinfo{articleno}{6} (\bibinfo{date}{jan} \bibinfo{year}{2019}), \bibinfo{numpages}{28}~pages.
\newblock
\href{https://doi.org/10.1145/3290319}{doi:\nolinkurl{10.1145/3290319}}


\bibitem[Birkedal et~al\mbox{.}(2012)]{DBLP:journals/corr/abs-1208-3596}
\bibfield{author}{\bibinfo{person}{Lars Birkedal}, \bibinfo{person}{Rasmus~Ejlers M{\o}gelberg}, \bibinfo{person}{Jan Schwinghammer}, {and} \bibinfo{person}{Kristian St{\o}vring}.} \bibinfo{year}{2012}\natexlab{}.
\newblock \showarticletitle{First steps in synthetic guarded domain theory: step-indexing in the topos of trees}.
\newblock \bibinfo{journal}{\emph{Log. Methods Comput. Sci.}} \bibinfo{volume}{8}, \bibinfo{number}{4} (\bibinfo{year}{2012}).
\newblock
\href{https://doi.org/10.2168/LMCS-8(4:1)2012}{doi:\nolinkurl{10.2168/LMCS-8(4:1)2012}}


\bibitem[Cartwright and Felleisen(1994)]{cartwright:1994}
\bibfield{author}{\bibinfo{person}{Robert Cartwright} {and} \bibinfo{person}{Matthias Felleisen}.} \bibinfo{year}{1994}\natexlab{}.
\newblock \showarticletitle{Extensible Denotational Language Specifications}. In \bibinfo{booktitle}{\emph{Theoretical {{Aspects}} of {{Computer Software}}}}, \bibfield{editor}{\bibinfo{person}{Masami Hagiya} {and} \bibinfo{person}{John~C. Mitchell}} (Eds.). \bibinfo{publisher}{Springer}, \bibinfo{address}{Berlin, Heidelberg}, \bibinfo{pages}{244--272}.
\newblock
\showISBNx{978-3-540-48383-0}
\href{https://doi.org/10.1007/3-540-57887-0_99}{doi:\nolinkurl{10.1007/3-540-57887-0_99}}


\bibitem[Danvy and Filinski(1989)]{danvy.filinski.1989}
\bibfield{author}{\bibinfo{person}{Olivier Danvy} {and} \bibinfo{person}{Andrzej Filinski}.} \bibinfo{year}{1989}\natexlab{}.
\newblock \bibinfo{booktitle}{\emph{A functional abstraction of typed contexts}}.
\newblock \bibinfo{type}{{T}echnical {R}eport}. \bibinfo{institution}{DIKU – Computer Science Department, University of Copenhagen}.
\newblock


\bibitem[de~Vilhena and Pottier(2021)]{vilhena-pottier}
\bibfield{author}{\bibinfo{person}{Paulo~Em\'{\i}lio de Vilhena} {and} \bibinfo{person}{Fran\c{c}ois Pottier}.} \bibinfo{year}{2021}\natexlab{}.
\newblock \showarticletitle{A separation logic for effect handlers}.
\newblock \bibinfo{journal}{\emph{Proc. ACM Program. Lang.}} \bibinfo{volume}{5}, \bibinfo{number}{POPL}, Article \bibinfo{articleno}{33} (\bibinfo{date}{jan} \bibinfo{year}{2021}), \bibinfo{numpages}{28}~pages.
\newblock
\href{https://doi.org/10.1145/3434314}{doi:\nolinkurl{10.1145/3434314}}


\bibitem[Dyvbig et~al\mbox{.}(2007)]{dyvbig:2007}
\bibfield{author}{\bibinfo{person}{R.~Kent Dyvbig}, \bibinfo{person}{Simon Peyton~Jones}, {and} \bibinfo{person}{Amr Sabry}.} \bibinfo{year}{2007}\natexlab{}.
\newblock \showarticletitle{A Monadic Framework for Delimited Continuations}.
\newblock  \bibinfo{volume}{17}, \bibinfo{number}{6} (\bibinfo{year}{2007}), \bibinfo{pages}{687--730}.
\newblock
Issue 6.
\showISSN{1469-7653, 0956-7968}
\href{https://doi.org/10.1017/S0956796807006259}{doi:\nolinkurl{10.1017/S0956796807006259}}


\bibitem[Felleisen and Friedman(1987)]{DBLP:conf/ifip2/FelleisenF87}
\bibfield{author}{\bibinfo{person}{Matthias Felleisen} {and} \bibinfo{person}{Daniel~P. Friedman}.} \bibinfo{year}{1987}\natexlab{}.
\newblock \showarticletitle{Control operators, the SECD-machine, and the {\(\lambda\)}-calculus}. In \bibinfo{booktitle}{\emph{Formal Description of Programming Concepts - {III:} Proceedings of the {IFIP} {TC} 2/WG 2.2 Working Conference on Formal Description of Programming Concepts - III, Ebberup, Denmark, 25-28 August 1986}}, \bibfield{editor}{\bibinfo{person}{Martin Wirsing}} (Ed.). \bibinfo{publisher}{North-Holland}, \bibinfo{pages}{193--222}.
\newblock


\bibitem[Frumin et~al\mbox{.}(2024)]{DBLP:journals/pacmpl/FruminTB24}
\bibfield{author}{\bibinfo{person}{Dan Frumin}, \bibinfo{person}{Amin Timany}, {and} \bibinfo{person}{Lars Birkedal}.} \bibinfo{year}{2024}\natexlab{}.
\newblock \showarticletitle{Modular Denotational Semantics for Effects with Guarded Interaction Trees}.
\newblock \bibinfo{journal}{\emph{Proc. {ACM} Program. Lang.}} \bibinfo{volume}{8}, \bibinfo{number}{{POPL}} (\bibinfo{year}{2024}), \bibinfo{pages}{332--361}.
\newblock
\href{https://doi.org/10.1145/3632854}{doi:\nolinkurl{10.1145/3632854}}


\bibitem[{Iris team}(2025)]{irisWWW}
\bibfield{author}{\bibinfo{person}{{Iris team}}.} \bibinfo{year}{2025}\natexlab{}.
\newblock \bibinfo{booktitle}{\emph{The {Iris} {Project} website and {Coq} development}}.
\newblock
\urldef\tempurl \url{https://iris-project.org/}
\showURL{\tempurl}


\bibitem[Jung et~al\mbox{.}(2018)]{DBLP:journals/jfp/JungKJBBD18}
\bibfield{author}{\bibinfo{person}{Ralf Jung}, \bibinfo{person}{Robbert Krebbers}, \bibinfo{person}{Jacques{-}Henri Jourdan}, \bibinfo{person}{Ales Bizjak}, \bibinfo{person}{Lars Birkedal}, {and} \bibinfo{person}{Derek Dreyer}.} \bibinfo{year}{2018}\natexlab{}.
\newblock \showarticletitle{Iris from the ground up: {A} modular foundation for higher-order concurrent separation logic}.
\newblock \bibinfo{journal}{\emph{J. Funct. Program.}}  \bibinfo{volume}{28} (\bibinfo{year}{2018}), \bibinfo{pages}{e20}.
\newblock
\href{https://doi.org/10.1017/S0956796818000151}{doi:\nolinkurl{10.1017/S0956796818000151}}


\bibitem[King(2021)]{GHC/Cont}
\bibfield{author}{\bibinfo{person}{Alexis King}.} \bibinfo{year}{2021}\natexlab{}.
\newblock \bibinfo{title}{Delimited continuation primops}.
\newblock \bibinfo{howpublished}{\url{https://github.com/ghc-proposals/ghc-proposals/blob/master/proposals/0313-delimited-continuation-primops.rst}}.
\newblock
\newblock
\shownote{Accessed: 2024-06-27}.


\bibitem[Kiselyov and Shan(2007)]{kiselyov:2007}
\bibfield{author}{\bibinfo{person}{Oleg Kiselyov} {and} \bibinfo{person}{Chung-chieh Shan}.} \bibinfo{year}{2007}\natexlab{}.
\newblock \showarticletitle{A {{Substructural Type System}} for {{Delimited Continuations}}}. In \bibinfo{booktitle}{\emph{Typed {{Lambda Calculi}} and {{Applications}}}}, \bibfield{editor}{\bibinfo{person}{Simona~Ronchi Della~Rocca}} (Ed.). \bibinfo{publisher}{Springer}, \bibinfo{address}{Berlin, Heidelberg}, \bibinfo{pages}{223--239}.
\newblock
\showISBNx{978-3-540-73228-0}
\href{https://doi.org/10.1007/978-3-540-73228-0_17}{doi:\nolinkurl{10.1007/978-3-540-73228-0_17}}


\bibitem[Koh et~al\mbox{.}(2019)]{KohEtAl:2019}
\bibfield{author}{\bibinfo{person}{Nicolas Koh}, \bibinfo{person}{Yao Li}, \bibinfo{person}{Yishuai Li}, \bibinfo{person}{Li-yao Xia}, \bibinfo{person}{Lennart Beringer}, \bibinfo{person}{Wolf Honor{\'e}}, \bibinfo{person}{William Mansky}, \bibinfo{person}{Benjamin~C. Pierce}, {and} \bibinfo{person}{Steve Zdancewic}.} \bibinfo{year}{2019}\natexlab{}.
\newblock \showarticletitle{From {{C}} to Interaction Trees: Specifying, Verifying, and Testing a Networked Server}. In \bibinfo{booktitle}{\emph{Proceedings of the 8th {{ACM SIGPLAN International Conference}} on {{Certified Programs}} and {{Proofs}}}} \emph{(\bibinfo{series}{{{CPP}} 2019})}. \bibinfo{publisher}{{Association for Computing Machinery}}, \bibinfo{address}{{New York, NY, USA}}, \bibinfo{pages}{234--248}.
\newblock
\showISBNx{978-1-4503-6222-1}
\href{https://doi.org/10.1145/3293880.3294106}{doi:\nolinkurl{10.1145/3293880.3294106}}


\bibitem[Krebbers et~al\mbox{.}(2017)]{DBLP:conf/popl/KrebbersTB17}
\bibfield{author}{\bibinfo{person}{Robbert Krebbers}, \bibinfo{person}{Amin Timany}, {and} \bibinfo{person}{Lars Birkedal}.} \bibinfo{year}{2017}\natexlab{}.
\newblock \showarticletitle{Interactive proofs in higher-order concurrent separation logic}. In \bibinfo{booktitle}{\emph{Proceedings of the 44th {ACM} {SIGPLAN} Symposium on Principles of Programming Languages, {POPL} 2017, Paris, France, January 18-20, 2017}}, \bibfield{editor}{\bibinfo{person}{Giuseppe Castagna} {and} \bibinfo{person}{Andrew~D. Gordon}} (Eds.). \bibinfo{publisher}{{ACM}}, \bibinfo{pages}{205--217}.
\newblock
\href{https://doi.org/10.1145/3009837.3009855}{doi:\nolinkurl{10.1145/3009837.3009855}}


\bibitem[Leijen(2014)]{DBLP:journals/corr/Leijen14}
\bibfield{author}{\bibinfo{person}{Daan Leijen}.} \bibinfo{year}{2014}\natexlab{}.
\newblock \showarticletitle{Koka: Programming with Row Polymorphic Effect Types}. In \bibinfo{booktitle}{\emph{Proceedings 5th Workshop on Mathematically Structured Functional Programming, MSFP@ETAPS 2014, Grenoble, France, 12 April 2014}} \emph{(\bibinfo{series}{{EPTCS}}, Vol.~\bibinfo{volume}{153})}, \bibfield{editor}{\bibinfo{person}{Paul~Blain Levy} {and} \bibinfo{person}{Neel Krishnaswami}} (Eds.). \bibinfo{pages}{100--126}.
\newblock
\href{https://doi.org/10.4204/EPTCS.153.8}{doi:\nolinkurl{10.4204/EPTCS.153.8}}


\bibitem[Leijen(2017)]{DBLP:conf/icfp/Leijen17}
\bibfield{author}{\bibinfo{person}{Daan Leijen}.} \bibinfo{year}{2017}\natexlab{}.
\newblock \showarticletitle{Structured asynchrony with algebraic effects}. In \bibinfo{booktitle}{\emph{Proceedings of the 2nd {ACM} {SIGPLAN} International Workshop on Type-Driven Development, TyDe@ICFP 2017, Oxford, UK, September 3, 2017}}, \bibfield{editor}{\bibinfo{person}{Sam Lindley} {and} \bibinfo{person}{Brent~A. Yorgey}} (Eds.). \bibinfo{publisher}{{ACM}}, \bibinfo{pages}{16--29}.
\newblock
\href{https://doi.org/10.1145/3122975.3122977}{doi:\nolinkurl{10.1145/3122975.3122977}}


\bibitem[Materzok and Biernacki(2011)]{materzok:2011}
\bibfield{author}{\bibinfo{person}{Marek Materzok} {and} \bibinfo{person}{Dariusz Biernacki}.} \bibinfo{year}{2011}\natexlab{}.
\newblock \showarticletitle{Subtyping Delimited Continuations}.
\newblock  \bibinfo{volume}{46}, \bibinfo{number}{9} (\bibinfo{year}{2011}), \bibinfo{pages}{81--93}.
\newblock
Issue 9.
\showISSN{0362-1340}
\href{https://doi.org/10.1145/2034574.2034786}{doi:\nolinkurl{10.1145/2034574.2034786}}


\bibitem[Pitts(2004)]{Pierce2004Advanced}
\bibfield{author}{\bibinfo{person}{Andrew~M. Pitts}.} \bibinfo{year}{2004}\natexlab{}.
\newblock \showarticletitle{Typed Operational Reasoning}.
\newblock In \bibinfo{booktitle}{\emph{Advanced Topics In Types And Programming Languages}}, \bibfield{editor}{\bibinfo{person}{Benjamin~C. Pierce}} (Ed.). \bibinfo{publisher}{The MIT Press}, Chapter~7, \bibinfo{pages}{245--289}.
\newblock
\showISBNx{0262162288}


\bibitem[Plotkin and Pretnar(2013)]{PlotkinPretnar:2013}
\bibfield{author}{\bibinfo{person}{Gordon~D. Plotkin} {and} \bibinfo{person}{Matija Pretnar}.} \bibinfo{year}{2013}\natexlab{}.
\newblock \showarticletitle{Handling {{Algebraic Effects}}}.
\newblock \bibinfo{journal}{\emph{Logical Methods in Computer Science}}  \bibinfo{volume}{Volume 9, Issue 4} (\bibinfo{date}{Dec.} \bibinfo{year}{2013}).
\newblock
\showISSN{1860-5974}
\href{https://doi.org/10.2168/LMCS-9(4:23)2013}{doi:\nolinkurl{10.2168/LMCS-9(4:23)2013}}


\bibitem[Silver et~al\mbox{.}(2023)]{SilverEtAl:2023}
\bibfield{author}{\bibinfo{person}{Lucas Silver}, \bibinfo{person}{Paul He}, \bibinfo{person}{Ethan Cecchetti}, \bibinfo{person}{Andrew~K Hirsch}, {and} \bibinfo{person}{Steve Zdancewic}.} \bibinfo{year}{2023}\natexlab{}.
\newblock \showarticletitle{Semantics for {{Noninterference}} with {{Interaction Trees}}}.
\newblock  (\bibinfo{year}{2023}).
\newblock


\bibitem[Sitaram and Felleisen(1991)]{sitaram:1991}
\bibfield{author}{\bibinfo{person}{Dorai Sitaram} {and} \bibinfo{person}{Matthias Felleisen}.} \bibinfo{year}{1991}\natexlab{}.
\newblock \showarticletitle{Models of Continuations without Continuations}. In \bibinfo{booktitle}{\emph{Proceedings of the 18th {{ACM SIGPLAN-SIGACT}} Symposium on {{Principles}} of Programming Languages}} (New York, NY, USA, 1991-01-03) \emph{(\bibinfo{series}{{{POPL}} '91})}. \bibinfo{publisher}{Association for Computing Machinery}, \bibinfo{pages}{185--196}.
\newblock
\showISBNx{978-0-89791-419-2}
\href{https://doi.org/10.1145/99583.99611}{doi:\nolinkurl{10.1145/99583.99611}}


\bibitem[Stepanenko et~al\mbox{.}(2025)]{DBLP:conf/esop/StepanenkoNFTB25}
\bibfield{author}{\bibinfo{person}{Sergei Stepanenko}, \bibinfo{person}{Emma Nardino}, \bibinfo{person}{Dan Frumin}, \bibinfo{person}{Amin Timany}, {and} \bibinfo{person}{Lars Birkedal}.} \bibinfo{year}{2025}\natexlab{}.
\newblock \showarticletitle{Context-Dependent Effects in Guarded Interaction Trees}. In \bibinfo{booktitle}{\emph{Programming Languages and Systems - 34th European Symposium on Programming, {ESOP} 2025, Held as Part of the International Joint Conferences on Theory and Practice of Software, {ETAPS} 2025, Hamilton, ON, Canada, May 3-8, 2025, Proceedings, Part {II}}} \emph{(\bibinfo{series}{Lecture Notes in Computer Science}, Vol.~\bibinfo{volume}{15695})}, \bibfield{editor}{\bibinfo{person}{Viktor Vafeiadis}} (Ed.). \bibinfo{publisher}{Springer}, \bibinfo{pages}{286--313}.
\newblock
\href{https://doi.org/10.1007/978-3-031-91121-7\_12}{doi:\nolinkurl{10.1007/978-3-031-91121-7\_12}}


\bibitem[Timany and Birkedal(2019)]{DBLP:journals/pacmpl/TimanyB19}
\bibfield{author}{\bibinfo{person}{Amin Timany} {and} \bibinfo{person}{Lars Birkedal}.} \bibinfo{year}{2019}\natexlab{}.
\newblock \showarticletitle{Mechanized relational verification of concurrent programs with continuations}.
\newblock \bibinfo{journal}{\emph{Proc. {ACM} Program. Lang.}} \bibinfo{volume}{3}, \bibinfo{number}{{ICFP}} (\bibinfo{year}{2019}), \bibinfo{pages}{105:1--105:28}.
\newblock
\href{https://doi.org/10.1145/3341709}{doi:\nolinkurl{10.1145/3341709}}


\bibitem[Timany et~al\mbox{.}(2024a)]{DBLP:journals/pacmpl/TimanyGB24}
\bibfield{author}{\bibinfo{person}{Amin Timany}, \bibinfo{person}{Arma{\"{e}}l Gu{\'{e}}neau}, {and} \bibinfo{person}{Lars Birkedal}.} \bibinfo{year}{2024}\natexlab{a}.
\newblock \showarticletitle{The Logical Essence of Well-Bracketed Control Flow}.
\newblock \bibinfo{journal}{\emph{Proc. {ACM} Program. Lang.}} \bibinfo{volume}{8}, \bibinfo{number}{{POPL}} (\bibinfo{year}{2024}), \bibinfo{pages}{575--603}.
\newblock
\href{https://doi.org/10.1145/3632862}{doi:\nolinkurl{10.1145/3632862}}


\bibitem[Timany et~al\mbox{.}(2024b)]{logical-approach-to-type-soundness}
\bibfield{author}{\bibinfo{person}{Amin Timany}, \bibinfo{person}{Robbert Krebbers}, \bibinfo{person}{Derek Dreyer}, {and} \bibinfo{person}{Lars Birkedal}.} \bibinfo{year}{2024}\natexlab{b}.
\newblock \showarticletitle{A Logical Approach to Type Soundness}.
\newblock \bibinfo{journal}{\emph{J. ACM}}  \bibinfo{volume}{71} (\bibinfo{year}{2024}).
\newblock
Issue 6.


\bibitem[{van den Berg} et~al\mbox{.}(2021)]{vandenBergEtAl:2021}
\bibfield{author}{\bibinfo{person}{Birthe {van den Berg}}, \bibinfo{person}{Tom Schrijvers}, \bibinfo{person}{Casper~Bach Poulsen}, {and} \bibinfo{person}{Nicolas Wu}.} \bibinfo{year}{2021}\natexlab{}.
\newblock \showarticletitle{Latent {{Effects}} for {{Reusable Language Components}}}. In \bibinfo{booktitle}{\emph{Programming {{Languages}} and {{Systems}}}} \emph{(\bibinfo{series}{Lecture {{Notes}} in {{Computer Science}}})}, \bibfield{editor}{\bibinfo{person}{Hakjoo Oh}} (Ed.). \bibinfo{publisher}{{Springer International Publishing}}, \bibinfo{address}{{Cham}}, \bibinfo{pages}{182--201}.
\newblock
\showISBNx{978-3-030-89051-3}
\href{https://doi.org/10.1007/978-3-030-89051-3_11}{doi:\nolinkurl{10.1007/978-3-030-89051-3_11}}


\bibitem[Wu et~al\mbox{.}(2014)]{WuEtAl:2014}
\bibfield{author}{\bibinfo{person}{Nicolas Wu}, \bibinfo{person}{Tom Schrijvers}, {and} \bibinfo{person}{Ralf Hinze}.} \bibinfo{year}{2014}\natexlab{}.
\newblock \showarticletitle{Effect Handlers in Scope}. In \bibinfo{booktitle}{\emph{Proceedings of the 2014 {{ACM SIGPLAN}} Symposium on {{Haskell}}}} \emph{(\bibinfo{series}{Haskell '14})}. \bibinfo{publisher}{{Association for Computing Machinery}}, \bibinfo{address}{{New York, NY, USA}}, \bibinfo{pages}{1--12}.
\newblock
\showISBNx{978-1-4503-3041-1}
\href{https://doi.org/10.1145/2633357.2633358}{doi:\nolinkurl{10.1145/2633357.2633358}}


\bibitem[Xia et~al\mbox{.}(2020)]{DBLP:journals/pacmpl/XiaZHHMPZ20}
\bibfield{author}{\bibinfo{person}{Li{-}yao Xia}, \bibinfo{person}{Yannick Zakowski}, \bibinfo{person}{Paul He}, \bibinfo{person}{Chung{-}Kil Hur}, \bibinfo{person}{Gregory Malecha}, \bibinfo{person}{Benjamin~C. Pierce}, {and} \bibinfo{person}{Steve Zdancewic}.} \bibinfo{year}{2020}\natexlab{}.
\newblock \showarticletitle{Interaction trees: representing recursive and impure programs in Coq}.
\newblock \bibinfo{journal}{\emph{Proc. {ACM} Program. Lang.}} \bibinfo{volume}{4}, \bibinfo{number}{{POPL}} (\bibinfo{year}{2020}), \bibinfo{pages}{51:1--51:32}.
\newblock
\href{https://doi.org/10.1145/3371119}{doi:\nolinkurl{10.1145/3371119}}


\bibitem[Zakowski et~al\mbox{.}(2021)]{ZakowskiEtAl:2021}
\bibfield{author}{\bibinfo{person}{Yannick Zakowski}, \bibinfo{person}{Calvin Beck}, \bibinfo{person}{Irene Yoon}, \bibinfo{person}{Ilia Zaichuk}, \bibinfo{person}{Vadim Zaliva}, {and} \bibinfo{person}{Steve Zdancewic}.} \bibinfo{year}{2021}\natexlab{}.
\newblock \showarticletitle{Modular, Compositional, and Executable Formal Semantics for {{LLVM IR}}}.
\newblock \bibinfo{journal}{\emph{Proceedings of the ACM on Programming Languages}} \bibinfo{volume}{5}, \bibinfo{number}{ICFP} (\bibinfo{date}{Aug.} \bibinfo{year}{2021}), \bibinfo{pages}{67:1--67:30}.
\newblock
\href{https://doi.org/10.1145/3473572}{doi:\nolinkurl{10.1145/3473572}}


\bibitem[Zhang et~al\mbox{.}(2021)]{ZhangEtAl:2021}
\bibfield{author}{\bibinfo{person}{Hengchu Zhang}, \bibinfo{person}{Wolf Honor{\'e}}, \bibinfo{person}{Nicolas Koh}, \bibinfo{person}{Yao Li}, \bibinfo{person}{Yishuai Li}, \bibinfo{person}{Li-Yao Xia}, \bibinfo{person}{Lennart Beringer}, \bibinfo{person}{William Mansky}, \bibinfo{person}{Benjamin Pierce}, {and} \bibinfo{person}{Steve Zdancewic}.} \bibinfo{year}{2021}\natexlab{}.
\newblock \showarticletitle{Verifying an {{HTTP Key-Value Server}} with {{Interaction Trees}} and {{VST}}}. In \bibinfo{booktitle}{\emph{12th {{International Conference}} on {{Interactive Theorem Proving}} ({{ITP}} 2021)}} \emph{(\bibinfo{series}{Leibniz {{International Proceedings}} in {{Informatics}} ({{LIPIcs}})}, Vol.~\bibinfo{volume}{193})}, \bibfield{editor}{\bibinfo{person}{Liron Cohen} {and} \bibinfo{person}{Cezary Kaliszyk}} (Eds.). \bibinfo{publisher}{{Schloss Dagstuhl \textendash{} Leibniz-Zentrum f\"ur Informatik}}, \bibinfo{address}{{Dagstuhl, Germany}}, \bibinfo{pages}{32:1--32:19}.
\newblock
\showISBNx{978-3-95977-188-7}
\showISSN{1868-8969}
\href{https://doi.org/10.4230/LIPIcs.ITP.2021.32}{doi:\nolinkurl{10.4230/LIPIcs.ITP.2021.32}}


\end{thebibliography}

\end{document}